\documentclass{article}
\usepackage{amsmath}
\usepackage{amsthm}
\usepackage{graphicx}
\usepackage{amsfonts}
\usepackage{amssymb}
\usepackage{enumerate}
\usepackage{color}
\usepackage{float}

\usepackage{natbib}

\newcommand{\A}{\mathbb{A} }
\newcommand{\EE}{\mathbb{E} }
\newcommand{\FF}{\mathbb{F} }
\newcommand{\PP}{\mathbb{P} }

\newcommand{\RR}{\mathbb{R} }

\newcommand{\cF}{\mathcal{F} }
\newcommand{\cH}{\mathcal{H} }

\def\balpha{\boldsymbol{\alpha}}
\def\bbeta{\boldsymbol{\beta}}

\def\bX{{\mathbf X}}
\def\bY{{\mathbf Y}}
\def\bZ{{\mathbf Z}}

\def\o{\overline}

\def\t{\tilde}

\usepackage{color}

\newcommand{\eps}{\varepsilon}

\newcommand{\be}{\begin{equation}}
\newcommand{\en}{\end{equation}}

\newtheorem{theorem}{Theorem}
\newtheorem{prop}{Proposition}
\newtheorem{lemma}{Lemma}

\newtheorem{remark}{Remark}

\newcommand{\ea}{\end{eqnarray}}
\newcommand{\ba}{\begin{eqnarray}}
\newcommand{\ean}{\end{eqnarray*}}
\newcommand{\ban}{\begin{eqnarray*}}

\title{SYSTEMIC RISK AND STOCHASTIC GAMES WITH DELAY }

\author{Ren\'e Carmona\thanks{ORFE, Bendheim Center for Finance, Princeton University, Princeton, NJ 08544. {\em
rcarmona@princeton.edu}. }
\and Jean-Pierre Fouque\thanks{Department of Statistics \& Applied Probability,
 University of California,
        Santa Barbara, CA 93106-3110, {\em fouque@pstat.ucsb.edu}. Work  supported by NSF grant DMS-1409434.}
        \and Seyyed Mostafa Mousavi\thanks{Department of Statistics \& Applied Probability,
 University of California,
        Santa Barbara, CA 93106-3110, {\em mousavi@pstat.ucsb.edu}}
        \and Li-Hsien Sun\thanks{Institute of Statistics, National Central University, Chung-Li, Taiwan 32001, {\em lihsiensun@ncu.edu.tw}.
        Work supported by MOST-103-2118-M-008-006-MY2.}}
\date{\today}

\begin{document}

\maketitle

\begin{abstract}
We propose a  model of inter-bank lending and borrowing which takes into account clearing debt obligations. The evolution of log-monetary reserves of $N$ banks is described by coupled diffusions driven by controls with delay in their drifts. Banks are minimizing their finite-horizon objective functions which take into account a quadratic cost for lending or borrowing and a linear incentive to borrow if the reserve is low or lend if the reserve is high relative to the average capitalization of the system. As such, our problem is an $N$-player linear-quadratic stochastic differential game with delay.   An open-loop Nash equilibrium is obtained using a system of fully coupled forward and advanced backward stochastic differential equations. We then describe how the delay affects liquidity and systemic risk characterized by a large number of defaults. We also derive a close-loop Nash equilibrium using an HJB approach.

\end{abstract}


\section{Introduction}\label{Intro}

In  \cite{R.Carmona2013}, we 
proposed a   stochastic game model of inter-bank lending and borrowing  where banks  borrow from or lend to a central bank with no obligation to pay back their loans and no gain from lending. The main finding was that in equilibrium, the central bank is acting as a clearing house,  liquidity is created, thus leading to a more stable system. Systemic risk was analyzed as in \cite{Fouque-Sun} in the case of a linear model without control. Systemic risk being characterized as the rare event of a large number of defaults occurring when the average capitalization reaches a prescribed  level, the conclusion was that inter-bank lending and borrowing leads to stability through a flocking effect.
For this type of interaction without control, we also refer to \cite{Fouque-Ichiba,  Garnier-Mean-Field} and \cite{GarnierPapanicolaouYang}.

In order to make the toy model of \cite{R.Carmona2013} more realistic, we introduce delay in the controls.  This forces banks to take responsibility for past lending and borrowing. In this paper, the evolution of the log-monetary reserves of $N$ banks is described by a system of delayed stochastic differential equations, and banks try to minimize their costs or maximize their profits by controlling the rate of borrowing or lending. They interact via the average capitalization meaning that banks consider  this average as a critical level to determine borrowing from or lending to the central bank. 

We identify open-loop Nash equilibria  by solving fully coupled forward and \emph{advanced} backward stochastic differential equations (FABSDEs) introduced by  \cite{Peng2009}. Our conclusion is that the  new effect created by the need to \emph{pay back} or \emph{receive refunds} due to the presence of the delay in the controls, reduces the liquidity observed in the case without delay. However, despite these quantitative differences,  the central bank is still acting as a clearing house.  A closed-loop Nash equilibrium to this stochastic game with delay is derived from the Hamilton-Jacobi-Bellman (HJB) equation approach   using the results in \cite{F.Gozzi2009} and we provide a verification Theorem.


For a general introduction to BSDEs, stochastic control and stochastic differential games without delay, we refer to the recent monograph \cite{CarmonaSIAM2016}.
Stochastic control problems with delay have been studied from various points of view. When the delay only appears in the state variable, solutions to delayed optimal control problems were derived from variants of the Pontryagin-Bismut-Bensoussan stochastic maximum principle. See for instance  \cite{B.Oksendal2001} and \cite{ B.Oksendal2011}. Alternatively, in order to use dynamic programming, \cite{Larssen2002} and \cite {B.Larssen2003} reduce the system with delay to a finite-dimension problem, but still the delay does not appear in the control like in the case we want to study.

The general case of stochastic optimal control of stochastic differential equations with delay both in the state and the control is studied using an infinite-dimensional  HJB equation in \cite{Gozzi2006}, and 
\cite{F.Gozzi2009}. The case with pointwise delayed control  is studied in   \cite{F.Gozzi2015}.  The general stochastic control problem in the case of delayed states and controls both appearing in the forward equation  is studied in \cite{L.Chen2011}, \cite{L.Chen2012} and  \cite{Xu2013} by using the forward and advanced backward stochastic equations. Linear-Quadratic mean field Stackelberg games with delay and with a major player and many small players are studied in \cite{P.Yam2016}.

The typical problem studied in this paper can be described as follows. The dynamics of the log-monetary reserves of $N$ banks are given by the following coupled diffusion processes $X^i_t$, $i=1,\cdots,N$,
\ba\label{Xi-1}
dX^{i}_t &=&\left(\alpha^{i}_t-\alpha^i_{t-\tau}\right)dt+\sigma d{W}^{i}_t,\quad 0\leq t\leq T,
\ea
where ${W}^i_t$, $i=1,\cdots,N$ are independent standard Brownian motions, and the rate of borrowing or lending $\alpha^{i}_t$ represents the control exerted by bank $i$ on the system. In this example, we use the simplest possible form of delay, the delayed control $\alpha^i_{t-\tau}$ corresponding to repayments after a fixed time $\tau$ such that $0\leq \tau\leq T$. We shall use deterministic initial conditions given by 
\be
X^i_0=\xi^i, 
\qquad\text{and}\qquad
\alpha^i_t=0,\quad t\in[-\tau,0).
\en
For simplicity, we assume that the banks have the same volatility $\sigma >0$. In what follows we use the notations $X=(X^1,\cdots,X^N)$, $x=(x^1,\cdots,x^N)$, $\alpha=(\alpha^1,\cdots,\alpha^N)$, and $\overline{x}=\frac{1}{N}\sum_{i=1}^{N}x^{i}$.

 
Before concentrating on the specific case \eqref{Xi-1}, we prove a dedicated version of the sufficient condition of the Pontryagin stochastic maximum principle for a more general class of models for which the dynamics of the states are given by stochastic differential equations of the form:
\ba\label{Xi-1-general}
dX^{i}_t &=&\left(\int_0^\tau\alpha^i_{t-s}\theta(ds)\right)dt+\sigma d{W}^{i}_t,\quad 0\leq t\leq T,
\ea 
where $\theta$ is a nonnegative measure on $[0,\tau]$. The special case (\ref{Xi-1}) corresponds to $\theta=\delta_0-\delta_\tau$.

Bank $i$ chooses its own strategy $\alpha^i$ in order to minimize its objective function of the form:
\ba\label{Ji-1}
J^{i}(\alpha)&=&\EE\bigg\{\int_{0}^{T}f_i(X_t,\alpha^i_t)
dt+g_i(X_T)\bigg\}.
\ea
In this paper, we concentrate on the running and terminal cost functions used in \cite{R.Carmona2013}, namely:
\be\label{fi}
f_i(x,\alpha^i)
=\frac{(\alpha^i)^2}{2}-q\alpha^i(\overline{x}-x^i)
+\frac{\epsilon}{2}(\overline{x}-x^i)^2,\quad q\geq 0,\quad \epsilon> 0,
\en
and
\be\label{gi}
g_i(x)=\frac{c}{2}\left(\overline{x}-x^{i}\right)^2,\quad c\geq 0,
\en
with $q^2< \epsilon$ so that $f_i(x,\alpha)$ is convex in $(x,\alpha)$. 
Note that the case $\tau>T$ corresponds to no repayment and therefore no delay in the equations,. The case $\tau=0$ corresponds to the case with no control and therefore no lending or borrowing. The term $q\alpha^i(\overline{x}-x^i)$ in the objective function (\ref{fi}) is an incentive to lend or borrow from a central bank which in this model does not make any decision and simply provides liquidity. However,  we know that in the case with no delay (\cite{R.Carmona2013}), in equilibrium, the central bank acts as a clearing house. We will see in Section \ref{FI} that this is still the case with delay.

\vskip 4pt
The paper is organized as follows. In Section \ref{MFG-Systemic-Risk}, we briefly review the model without delay presented in \cite{R.Carmona2013}. 
The analysis of the stochastic differential games with delay is presented in Section \ref{Systemic-Risk-SDDE}
where we derive an exact open-loop Nash equilibrium using the FABSDE approach. In the process, we derive the \emph{clearing house} role of the central bank in Remark \ref{re:clearing}. 
Section \ref{sec:HJB} is devoted to the derivation of a closed-loop equilibrium using an infinite-dimensional  HJB equation approach with pointwise delayed control presented in \cite{F.Gozzi2015}.
In Section \ref{caratheodory}, we provide a verification Theorem. The effect of delay in term of  financial implication is discussed in Section \ref{FI} where the main finding is that the introduction of delay in the model does not change the fact that in equilibrium, the central bank acts as a clearing house. However, liquidity is affected by the delay time.

\section{Stochatic  Games and Systemic Risk}\label{MFG-Systemic-Risk}\label{sec:nodelay}

The aim of this section is to briefly review the model of inter-bank lending or borrowing without delay studied in \cite{R.Carmona2013}.  
It is described by the model presented in the previous section but with $\tau>T$ so that the delay term $\alpha^i_{t-\tau}$ in (\ref{Xi-1}) is simply zero. The setup (\ref{Ji-1},\ref{fi},\ref{gi}) of the stochastic game remains the same.

The open-loop problem consists in searching for an equilibrium among strategies $\{\alpha^i_t, i=1,\cdots,N\}$ which are adapted processes satisfying some integrability property such as $\EE\left(\int_0^T|\alpha^i_t|dt\right)<\infty$. 
The Hamiltonian for bank $i$ is given by
\begin{align}
\label{Hamiltonian}
H^i(x,y^i,\alpha)
&=
\sum_{k=1}^N\alpha^ky^{i,k}+\frac{(\alpha^i)^2}{2}
-q\alpha^i(\o x-x^i)+\frac{\epsilon}{2}(\o x-x^i)^2,
\end{align}
where $y^i=(y^{i,1},\cdots,y^{i,N}),\, i=1,\cdots,N$ are the adjoint variables.

For a given $\alpha=(\alpha^i)_ {i=1,\cdots,n}$, the controlled forward dynamics of the states $X^i_t$ are given by (\ref{Xi-1}) without the delay term and with initial conditions $X^i_0=\xi^i$. The adjoint processes $Y^i_t=(Y_t^{i,j};\,j=1,\cdots,N)$ and $Z^i_t=(Z_t^{i,j,k};\; j=1,\cdots,N,\; k=1,\cdots,N)$ for $ i=1,\cdots,N$ are defined as the solutions of the backward stochastic differential equations (BSDEs):
\be
\label{ol-dYij}
dY_t^{i,j}=-\partial_{x^j}H^i(X_t,Y^{i}_t,\alpha_t)dt+\sum_{k=1}^NZ^{i,j,k}_tdW^k_t
\en
with terminal conditions $Y^{i,j}_T=\partial_{x^j}g_i(X_T)$ for $i,j=1,\cdots,N$ where $g_i$ is given by (\ref{gi}). For each admissible strategy profile $\alpha=(\alpha^i)_ {i=1,\cdots,n}$, standard existence and uniqueness results for BSDEs apply and the existence of the adjoint processes is guaranteed. Note that from (\ref{Hamiltonian}), we have
\ban
\partial_{x^j}H^i=-q\alpha^i(\frac{1}{N}-\delta_{i,j})+\epsilon(\overline{x}-x^i)(\frac{1}{N}-\delta_{i,j}).
\ean
The necessary condition of the Pontryagin stochastic maximum principle suggests that one minimizes the Hamiltonian $H^i$  with respect to $\alpha^i$ which gives:
\be\label{alpha-open-loop}
\hat\alpha^i=-y^{i,i}+q(\o x-x^i).
\end{equation}
With this choice for the controls $\alpha^i$, the forward  equation  becomes coupled with the backward equation (\ref{ol-dYij}) to form a forward-backward coupled system. In the present linear-quadratic  case, we make the 
ansatz
\be\label{ansatz-open-loop}
Y_t^{i,j}=\phi_t(\frac{1}{N}-\delta_{i,j})(\o X_t-X^i_t),
\en
for some deterministic scalar function $\phi_t$ satisfying the terminal condition $\phi_T=c$. 
Using this ansatz, the backward equations (\ref{ol-dYij}) become
\be\label{ol-dYij2}
dY_t^{i,j}=(\frac{1}{N}-\delta_{i,j})(\o X_t-X^i_t)\left[q(1-\frac{1}{N})\phi_t-(\epsilon-q^2)\right]dt+
\sum_{k=1}^NZ^{i,j,k}_tdW^k_t.
\en
Using (\ref{alpha-open-loop}) and (\ref{ansatz-open-loop}), the forward equation  becomes
\be\label{ol-dXi2}
dX^i_t=\left[q+(1-\frac{1}{N})\phi_t\right](\o X_t-X^i_t)dt+\sigma  dW^i_t.
\en
Differentiating the ansatz (\ref{ansatz-open-loop}) and identifying with the Ito's representation (\ref{ol-dYij2}), one obtains from the martingale terms the deterministic adjoint variables
$$
Z^{i,j,k}_t=\phi_t\sigma(\frac{1}{N}-\delta_{i,j})(\frac{1}{N}-\delta_{i,k})\,\mbox{for }\,\, k=1,\cdots,N,
$$
and from the drift terms  that the function $\phi_t$ must satisfy the scalar Riccati equation
  \be\label{olRiccati}
 \dot\phi_t=2q(1-\frac{1}{2N})\phi_t+(1-\frac{1}{N})\phi_t^2-(\epsilon-q^2),
 \en
with the terminal condition $\phi_T=c$. The explicit solution is given in \cite{R.Carmona2013}.
 Note that the form (\ref{alpha-open-loop}) of the control $\alpha^i_t$,  and the ansatz (\ref{ansatz-open-loop}) combine to give:
 \be\label{ol-alphai}
 \alpha^i_t=\left[q+(1-\frac{1}{N})\phi_t\right](\o X_t-X^i_t),
 \en
so that, in this equilibrium, the forward equations become
\be\label{X-nodelay1}
dX^i_t=\left(q+(1-\frac{1}{N})\phi_t\right)(\overline X_t-X^i_t)dt+\sigma d{W}^{i}_t.
\en
Rewriting  $(\overline X_t-X^i_t)$ as $\frac{1}{N}\sum_{j=1}^{N}(X^j_t-X^i_t)$, we see that the central bank is simply acting as a clearing house. From the form (\ref{X-nodelay1}), we observe that the $X^i$'s are mean-reverting to the average capitalization given by 
\[
d\o X_t=\frac{\sigma}{N}\sum_{j=1}^Nd{W}^j_t,\quad \o X_0=\frac{1}{N}\sum_{j=1}^N\xi^j.
\]
In \cite{Fouque-Sun}, we identified the systemic event as $$\left\{\min_{0\leq t\leq T}(\o X_t-\o X_0)\leq D\right\}$$ and we computed its probability
\ba\label{Prob-systemic-event}
\PP\left(\min_{0\leq t\leq T}(\o X_t-\o X_0)\leq D\right) &=& 2\Phi\left(\frac{D\sqrt{N}}{\sigma\sqrt{T}}\right),
\ea
where $\Phi$ is the ${\cal N}(0,1)$-cdf. This systemic risk probability is  exponentially small of  order $\exp(-D^2N/(2\sigma^2T))$ as in the large deviation estimate.

\section{Stochastic Games with Delay}
\label{Systemic-Risk-SDDE}
Most often, a tailor made version of the stochastic maximum principle is used as a workhorse to construct open loop Nash equilibria for stochastic differential games. Here, we provide such a tool in a more general set up than used in the paper because we believe that this result is of independent interest on its own. We then specialize it to the model considered for systemic risk in Section \ref{sec:example}.

\subsection{The Model}

We work with a finite horizon $T>0$. Recall that we denote by $\tau>0$ the delay length. As explained in the introduction, the delay is implemented with a (signed) measure $\theta$ on $[0,\tau]$, and in the case of interest, we shall use the particular case $\theta=\delta_0 - \delta_\tau$.
All the stochastic processes are defined on a probability space $(\Omega,\cF,\PP)$ equipped with a right continuous filtration $\FF=(\cF_t)_{0\le t\le T}$. The state and control processes are denoted by $\bold{X}=(X_t)_{0\le t\le T}$ and $\bold{\alpha}=(\alpha_t)_{0\le t\le T}$.
They are progressively measurable processes with values in $(\RR^d)^N$ and a closed convex subset $A$ of $(\RR^d)^N$ respectively.
They are linked by the dynamical equation:
\begin{equation}
\label{fo:control_equation}
dX_t=<\alpha_{[t]}, \theta> dt +\sigma dW_t
\end{equation}
where $\bold{W}
=(W_t)_{0\le t\le T}$ is a $(d\times N)$-dimensional $\FF$-Brownian motion, $\sigma$ is a positive constant or a matrix. We use the notation $\alpha_{[t]}=\alpha_{[t-\tau,t]}$ for the restriction of the path of $\alpha$ to the interval $[t-\tau,t]$. By convention, and unless specified otherwise, we extend functions defined on the interval $[0,T]$ to functions on $[-\tau,T+\tau]$ by setting them equal to $0$ outside the interval $[0,T]$.
Also, we use the bracket notation $<f, \theta>$ to denote the integral $\int_0^\tau f(s)\theta(ds)$.

\vskip 4pt
We assume that the dynamics of the state $X_t$ of the system are given by a stochastic differential equation
\eqref{fo:control_equation} which we can rewrite in coordinate form if we denote by $X^i_t$ the $N$ components of $X_t$, in which case we can interpret $X^i_t$ as the private state of player $i$:
\ba\label{fo:Xi_general}
dX^{i}_t &=&\left(\int_0^\tau\alpha^i_{t-s}\theta(ds)\right)dt+\sigma d{W}^{i}_t,\quad 0\leq t\leq T,
\ea 
where the components ${W}^i_t$, $i=1,\cdots,N$ of $W_t$ are independent standard Wiener processes, and the component processes $(\alpha^{i}_t)_{t\ge 0}$ can be interpreted as the strategies used by the individual players. As explained in the introduction,
 $\theta$ is a nonnegative measure on $[0,\tau]$ implementing the impact of the delay on the dynamics. 
Recall that the special case of interest corresponds to $\theta=\delta_0-\delta_\tau$.
We assume the initial conditions: 
\be
X^i_0=\xi^i, 
\qquad\text{and}\qquad
\alpha^i_t=0,\quad t\in[-\tau,0).
\en
The assumptions that the various states have the same volatility $\sigma >0$ and the delay measure $\theta$ is the same for all the players are only made for convenience. These symmetry properties are important to derive mean field limits, but they are not really needed when we deal with finitely many players.
The objective function of player $i$ is given by (\ref{Ji-1}) which we repeat here:
\ban
J^{i}(\alpha)&=&\EE\bigg\{\int_{0}^{T}f_i(X_t,\alpha^i_t)
dt+g_i(X_T)\bigg\}.
\ean
For the sake of simplicity, we assume that the cost $f_i$ to player $i$ depends only upon the control $\alpha^i_t$ of player $i$, and not on the controls $\alpha^j_t$ for $j\ne i$ of the other players.
In the case of games with mean field interactions, the cost functions are often of the form $f_i(x,\alpha)=f(x^i,\overline{x},\alpha)$ and $g_i(x)=g(x^i,\overline{x})$,
as in the particular case of the systemic risk model studied in this paper where:
\be
f_i(x,\alpha^i)=f(x^i,\overline{x},\alpha^i)
=\frac{(\alpha^i)^2}{2}-q\alpha^i(\overline{x}-x^i)
+\frac{\epsilon}{2}(\overline{x}-x^i)^2,\nonumber
\en
for $q\geq 0$ and $\epsilon>0$ as in (\ref{fi}), and:
\be
g_i(x)=g(x^i,\overline{x})=\frac{c}{2}\left(\overline{x}-x^{i}\right)^2,\quad c\geq 0,\nonumber
\en
as in (\ref{gi}) and with $q^2< \epsilon$ to make sure that $f_i(x,\alpha)$ is convex in $(x,\alpha)$. 
Next, we introduce the system of adjoint equations.

\subsection{The Adjoint Equations}
For each player $i$ and each given admissible control $\balpha^i
=(\alpha^i_t)_{0\le t\le T}$ for player $i$, 
we define the adjoint equation for player $i$ as the Backward Stochastic Differential Equation (BSDE):
\begin{equation}
\label{fo:games_adjoints}
d Y^i_t=-\partial_x f_i(X_t,\alpha^i_t) dt + Z^i_t dW_t, \qquad 0\le t\le T
\end{equation}
with terminal condition $Y^i_T=\partial_x g_i(X_T)$, and we call the processes $\bY^i=(Y^i_t)_{0\le t\le T}$ and $\bZ^i
=(Z^i_t)_{0\le t\le T}$ the  adjoint processes corresponding to the strategy $\balpha^i=(\alpha^i_t)_{0\le t\le T}$ of player $i$.
Notice that each $\bY^i$ has the same dimension as $\bX$, namely $N\times d$ if $d$ is the dimension of each individual player private state $X^i_t$, while each $\bZ^i$ has dimension $N^2\times d$.
Accordingly, we shall use the notation $Y^i_t=(Y^{i,j}_t)_{j=1,\cdots,N}$ where each $Y^{i,j}_t$ has the same dimension $d$ as each of the private states $X^j_t$, and similarly, $Z^i_t=(Z^{i,j,k}_t)_{j,k=1,\cdots,N}$. In the application of interest to us in this paper we have $d=1$.

\vskip 4pt
As before, the following notation will turn out to be helpful.
If $\bold{Y}
=(Y_t)_{0\le t\le T}$ is a progressively measurable process (scalar or multivariate) with continuous sample paths, we denote by $\tilde{\bold{Y}}=(\t Y_t)_{0\le t\le T}$
the process defined by:
$$
\t Y_t=\EE\bigg[\int_0^\tau Y_{t+s}\theta(ds)\;\big\vert \cF_t\bigg]=\int_0^\tau\EE[Y_{t+s}\vert \cF_t]\;\theta(ds),\qquad 0\le t\le T.
$$
Moreover, for each $t\in[0,T]$, $x\in (\RR^{d})^N$ and $y\in\RR^d$, we denote by $\hat\alpha^i(x,y)$ any $\alpha\in \RR^d$ satisfying:
\begin{equation}
\label{fo:alphahat}
\partial_\alpha f_i(x,\alpha) = -y.
\end{equation}
Under specific assumptions the implicit function theorem will provide existence of $\hat\alpha_i$, and regularity properties of this function with respect to the variables $x$  and $y$.

\subsection{Sufficient Condition for Optimality}

\begin{theorem}
\label{th:pontsufficient}
Let us assume that the cost functions $f_i$ are continuously differentiable in $(x,\alpha)\in(\RR^{d})^N\times \RR^d$, and $g_i$ are continuously differentiable on 
$(\RR^d)^N$ with partial derivatives of (at most) linear growth, and that:
\begin{itemize}\itemsep=-1pt
\item[(i)] the functions $g_i$  are convex;
\item[(ii)] the functions $(x,\alpha)\mapsto f_i(x,\alpha)$ are convex.
\end{itemize}
If  $\balpha=(\alpha^1_t,\cdots,\alpha^N_t)_{0\le t\le T}$ is an admissible adapted (open loop) strategy profile, and  $(\bX,\bY,\bZ)=\big(( {X_t^1},\cdots, {X_t^N}),( {Y_t^1},\cdots, {Y_t^N}),( {Z_t^1},\cdots, {Z_t^N})\big)$ are adapted process such that the dynamical equation \eqref{fo:control_equation} and the adjoint equations \eqref{fo:games_adjoints}
are satisfied for the controls $\alpha^i_t=\hat\alpha^i(X_t,\tilde{Y}^{i,i}_t)$, then
the strategy profile $\balpha=(\alpha^1_t,\cdots,\alpha^N_t)_{0\le t\le T}$ is an open loop Nash equilibrium.
\end{theorem}

\begin{proof}
We follow the proof given in \cite{CarmonaSIAM2016} in the case without delay. We fix $i\in\{1,\cdots,N\}$, a generic admissible control strategy $(\beta_t)_{0\le t\le T}$ for player $i$, and for the sake of simplicity, we denote by $X'$ the state $X_t^{(\hat \alpha^{-i},\beta)}$ controlled by the strategies $(\hat{\alpha}^{-i},\beta)$. The function $g_i$ being convex, almost surely, we have:
{\small
\begin{eqnarray}
\label{fo:gdiff}
&&g_i(X_T)-g_i(X'_T)\nonumber\\
&&\phantom{?}\le (X_T-X'_T) \cdot \partial_xg_i(X_T)\nonumber\\
&&\phantom{?}= (X_T-X'_T) \cdot Y^i_T\nonumber\\
&&\phantom{?}=\int_0^T (X_t-X'_t)\;d Y^i_t +\int_0^T  Y^i_t\;d( X_t-X'_t) \nonumber\\
&&\phantom{?}=-\int_0^T (X_t-X'_t) \cdot \partial_xf_i(X_t,\alpha^i_t)\; dt
+\int_0^T Y^i_t \cdot <\alpha_{[t]}-(\hat \alpha^{-i},\beta)_{[t]},\theta>\;dt \;+\:\text{martingale}\nonumber\\
&&\phantom{?}=-\int_0^T (X_t-X'_t) \cdot \partial_xf_i(X_t,\alpha^i_t)\; dt
+\int_0^T Y^{i,i}_t \cdot <\alpha^i_{[t]}-\beta_{[t]},\theta>\;dt \;+\:\text{martingale}.\nonumber
\end{eqnarray}
}
Notice that we can use the classical form of integration by parts is due to the fact that the volatilities of all the states are the same constant $\sigma$.
Taking expectations of both sides and plugging  the result into
$$
J^i(\balpha)-J^i((\balpha^{-i},\bbeta))=\EE\bigg\{\int_0^T[f_i(X_t,\alpha^i_t)-f_i(X'_t, \beta_t)]dt\bigg\}+\EE\{g_i(X_T)-g_i(X'_T)\},
$$
we get:
\begin{eqnarray}
\label{fo:Jdiff}
&&J^i(\balpha)-J^i((\balpha^{-i},\bbeta))\nonumber\\
&&\phantom{????}\le\EE\left\{\int_0^T[f_i( X_t,\alpha^i_t)-f_i(X'_t, \beta_t)]dt 
-\int_0^T (X_t-X'_t) \cdot \partial_xf_i(X_t,\alpha^i_t)\; dt\right\}\nonumber\\
&&\phantom{???????????}+\EE\left\{\int_0^T Y^{i,i}_t \cdot <\alpha^i_{[t]}-\beta_{[t]},\theta>\;dt \right\}\nonumber\\
&&\phantom{????}\le\EE\left\{\int_0^T[\alpha^i_t-\beta_t]\partial_\alpha f_i(X_t,\alpha^i_t)
+Y^{i,i}_t \cdot <\alpha^i_{[t]}-\beta_{[t]},\theta>\;dt \right\}.
\end{eqnarray}
Notice that:
\begin{equation*}
\begin{split}
\EE\bigg[\int_0^T Y^{i,i}_t\cdot <\alpha^i_{[t]}-\beta_{[t]},\theta> dt\bigg]
&=\EE\bigg[\int_0^\tau\bigg(\int_{-s}^{T-s}  Y^{i,i}_{t+s}[\alpha^i_{t}-\alpha^i_{t}]dt\bigg)\,\theta(ds)\bigg]\\
&=\int_0^\tau\int_{0}^{T} \EE[Y^{i,i}_{t+s}[\alpha^i_{t}-\beta_{t}]dt\,\theta(ds)]\\
&=\int_0^\tau\int_{0}^{T} \EE[\EE[Y^{i,i}_{t+s}|\cF_t][\alpha^i_{t}-\beta_{t}]dt\,\theta(ds)]\\
&=\EE\bigg[\int_0^\tau\int_{0}^{T}\bigg(\int_0^\tau \EE[Y^{i,i}_{t+s}|\cF_t]\theta(ds)\bigg)[\alpha^i_{t}-\beta_{t}]dt\,\bigg]\\
&=\EE\bigg[\int_0^T \widetilde{Y}_t^{i,i}\cdot [\alpha^i_{t}-\beta_{t}]dt\bigg].
\end{split}
\end{equation*}
Consequently:
\begin{equation*}
\begin{split}
J^i(\balpha)-J^i((\balpha^{-i},\bbeta))&\le \EE\left\{\int_0^T\bigg([\alpha^i_t-\beta_t]\partial_\alpha f_i(X_t,\alpha^i_t)
+ \widetilde{Y}_t^{i,i}\cdot [\alpha^i_{t}-\beta_{t}]\bigg)\;dt \right\}
\\
&=0
\end{split}
\end{equation*}
by definition \eqref{fo:alphahat} of $\hat\alpha(t,\hat X_t,\widetilde{Y}_t^{i,i})$.
\end{proof}

\subsubsection{Example}\label{sec:example}
We shall use the above result when $d=1$,  $\theta=\delta_0-\delta_{-\tau}$
so that $<\alpha_{[t]},\theta>=\int_0^\delta \alpha_{t-\tau}\,\theta(d\tau)=\alpha_t-\alpha_{t-\delta}$,  and the cost functions are given by \eqref{fi} and \eqref{gi}, namely:
$$
f_i(x,\alpha) = \frac12\alpha^2-q\alpha(\o x - x^i)+\frac{\epsilon}{2}(\o x- x^i)^2
$$
for some positive constants $q$ and $\epsilon$ satisfying $q<\epsilon^2$ which guarantees that 
the functions $f_i$ are convex.
Notice that relation \eqref{fo:alphahat} gives  $\hat\alpha^i(x,y)=-y-q(x^i-\o x)$. To derive
the adjoint equations we compute:
$$
\partial_{x^i} f_i(x,\alpha) = \bigl(1-\frac1N\bigr)[q\alpha+\epsilon(x^i-\o x)],
\qquad\text{and}\qquad
\partial_{x^j} f_i(x,\alpha) = -\frac1N[q\alpha+\epsilon(x^i-\o x)],
$$
for $j\ne i$. Accordingly, the system of forward and advanced backward equations identified in the above theorem reads:
{\small
\begin{equation}
\label{fo:FABSDE_general}
\begin{cases}
&dX^i_t=- <\widetilde{Y}_{[t]}^{i,i}+ q (X^i_{[t]}-\o X_{[t]}),\theta>dt +\sigma dW^i_t,\qquad i=1,\cdots,N\\
&dY^{i,j}_t=\bigl(\delta_{i,j}-\frac1N\bigr)[q\widetilde{Y}^{i,j}_t+(q^2-\epsilon)(X^i_t-\o X_t)]dt + \sum_{k=1}^NZ^{i,j,k}_t dW^k_t\qquad i,j=1,\cdots,N\\
\end{cases}
\end{equation}}
where we used the Kronecker symbol $\delta_{i,j}$ which is equal to $1$ if $i=j$ and $0$ if $i\ne j$. If we specialize this system to the case $\theta=\delta_0-\delta_\tau$, we have $\widetilde{Y}^{i,j}_t=Y^{i,j}_t - \EE[Y^{i,j}_{t+\tau}|\cF_t]$, so that the forward advanced-backward system reads:
{\small
\begin{equation}
\label{fo:FABSDE}
\begin{cases}
&dX^i_t=\bigl( -Y^{i,i}_t+Y^{i,i}_{t-\tau} +\EE[Y^{i,i}_{t+\tau}|\cF_t] - \EE[Y^{i,i}_{t}|\cF_{t-\tau}]\\
&\hskip 75pt - q[ X^i_t-X^i_{t-\tau} -\o X_t+\o X_{t-\tau}]\bigr)dt +\sigma dW^i_t,\qquad i=1,\cdots,N\\
&dY^{i,j}_t=\bigl(\delta_{i,j}-\frac1N\bigr)[q Y^{i,j}_t- q\EE[Y^{i,j}_{t+\tau}|\cF_t]+(q^2-\epsilon)(X^j_t-\o X_t)]dt + \sum_{k=1}^NZ^{i,j,k}_t dW^k_t\\
&\hskip 255pt  i,j=1,\cdots,N.\\
\end{cases}
\end{equation}
}

The version of the stochastic maximum principle proved in Theorem \ref{th:pontsufficient} reduces the problem of the existence of Nash equilibria for the system, to the solution of forward anticipated-backward stochastic differential equation. The following result can be used to resolve the existence issue but first we make the following remark which is key in term of financial interpretation.

\begin{remark}[\textbf{Clearing House Property}] 
\label{re:clearing}
In the present situation, in contrast with the case without delay presented in Section \ref{MFG-Systemic-Risk}, we will not be able to derive explicit formulas for the equilibrium optimal strategies such as (\ref{ol-alphai}). 
However, it is remarkable to see that the {\it clearing house property} $\sum \alpha^i=0$ still holds. Indeed, setting $i=j$ in \eqref{fo:FABSDE_general} and summing over $N$ to derive an equation for  $\overline{Y}_t=\frac{1}{N}\sum_{i=1}^NY^{i,i}_t$ and
$\overline Z^{k}_t=\frac{1}{N}\sum_{i=1}^NZ^{i,i,k}_t$, we find:
\be
\nonumber d\overline{Y}_t=-\left(\frac{1}{N}-1\right)
q \widetilde{\overline{Y}_t}dt
+\sum_{k=1}^N\overline Z^{k}_t dW^k_t,\quad t\in[0,T],
\en
with terminal condition $\overline{Y}_t=0$ for $t\in[T,T+\tau]$. This equation admits the unique solution: 
\[
\overline{Y}_t=0,\quad t\in [0,T+\tau], \quad \mbox{and}\quad \overline Z^{k}_t=0,\, k=1,\cdots, N, \, t\in [0,T].
\]
and as a result,
\be
\label{optimal-control-bar}
\overline{\hat\alpha}_t=-\widetilde{\overline{Y}_t}=0.
\en
\end{remark}

\vskip 4pt
In what follows, on the top of $q^2<\epsilon$, we further assume that
\ba\label{condition}
q^2(1-\frac{1}{2N})^2\leq \epsilon(1-\frac{1}{N}),
\ea
which is satisfied for $N$ large enough, or $q$ small enough.

\begin{theorem} 
\label{th:FABSDE_existence} 
The FABSDE \eqref{fo:FABSDE} has a unique solution.
\end{theorem}

\begin{remark}
While this theorem gives existence of open loop Nash equilibria for the model, it is unlikely that uniqueness holds. However,  the cost functions $f_i$ and $g_i$ depending only upon $x^i$ and $\overline{x}$, one could consider the mean field game problem corresponding to the limit $N\to\infty$, and in this limiting regime, it is likely that the strict convexity of the cost functions could be used to prove some form of uniqueness of the solution of the equilibrium problem.
\end{remark}

\begin{proof}
We first solve the system considering only the case $j=i$. Once this is done, we should be able to inject the process $X_t=(X^1_t,\cdots,X^N_t)$ so obtained into the equation for $dY^{i,j}_t$ for $j\ne i$, and solve this advanced equation with random coefficients.

\vskip 2pt
Summing over $i=1,\cdots,N$ the equations for $X^{i}$ in \eqref{fo:FABSDE_general}, using the clearing house property of Remark \ref{re:clearing}, and denoting $\overline{\xi}=\frac{1}{N}\sum_{i=1}^N\xi^i$ give
\ba
\label{eq:Xbar} \overline{X}_t&=&\overline{\xi}+ \frac{\sigma}{N}\sum_{i=1}^NW^i_t,\quad t\in [0,T].
\ea
Therefore, without loss of generality, we can work with the ``centered" variables
$X_t^{i,c}=X_t^i-\overline{X}_t$, $Y_t^{i,i,c}=Y_t^{i,i}-\overline{Y}_t=Y_t^{i,i}$,  and $Z_t^{i,i,k,c}=Z^{i,i,k}_t-\overline Z^{k}_t=Z^{i,i,k}_t$ which must satisfy the system:
\begin{equation}
\label{fo:FABSDE_centered}
\begin{cases}
&dX_t^{i,c}=- <\widetilde{Y}_{[t]}^{i,i}+ q X_{[t]}^{i,c},\theta>dt +\sigma \sum_{k=1}^N\bigl(\delta_{i,k}-\frac1N\bigr)dW^k_t,\\
&dY^{i,i}_t=\bigl(1-\frac1N\bigr)[q\widetilde{Y}^{i,i}_t+(q^2-\epsilon)X_t^{i,c}]dt + \sum_{k=1}^NZ^{i,i,k}_t dW^k_t\end{cases}
\end{equation}
with $X_0^{i,c}=\xi^{i,c}:=\xi^i-\overline\xi$, $Y_T^{i,i}=-c\left(\frac{1}{N}-1\right) X_T^{i,c}$, and $Y_t^{i,i}=0$ for $t\in(T,T+\tau]$ for $i=1,\cdots,N$.
We solve this system by extending the continuation method (see for example \cite{Peng1999} and 
\cite{Peng2009}) to the case of stochastic games.
We consider a system which is written as a perturbation of the previous one without delay. Since we now work with $i\in\{1,\cdots,N\}$ fixed, we drop the exponent $i$ from the notation for the sake of readability of the formulas.
\begin{equation}
\label{fo:Peng_Wu_system}
\begin{cases}
&dX_t^{\lambda}=\bigl[-(1-\lambda)Y_t^{\lambda}-\lambda <\widetilde{Y}_{[t]}^{\lambda}+ q X_{[t]}^{\lambda},\theta>+\phi_t\bigr]dt\\
&\hskip 135pt  + \sum_{k=1}^N
\bigl[-(1-\lambda)Z_t^{k,\lambda}+\lambda\sigma\bigl(\delta_{i,k}-\frac1N\bigr)+\psi_t^{k}\bigr]dW^k_t,\\
&dY^{\lambda}_t=\bigl[ -(1-\lambda)X_t^{\lambda} +\lambda\bigl(1-\frac1N\bigr)[q\widetilde{Y}^{\lambda}_t+(q^2-\epsilon)X_t^{\lambda}]+r_t\bigr]dt 
+ \sum_{k=1}^NZ_t^{k,\lambda} dW^k_t\end{cases}
\end{equation}
with initial condition $X^{\lambda}_0=\xi^{i,c}$ and terminal condition $Y_T^{\lambda}=(1-\lambda)X_T^{\lambda} -\lambda c\left(\frac{1}{N}-1\right)X_T^{\lambda}+\zeta^{i,i}$ and
$Y_t^{\lambda}=0$ for $t\in(T,T+\tau]$ in the case of $c>0$, and $Y_T^{\lambda}= \zeta^{i,i}$
and  $Y_t^{\lambda}=0$ for $t\in(T,T+\tau]$
in the case of $c=0$. 

\vskip 2pt
Here (recall that $i$ is now fixed), $\phi_t$, $\psi^{k}_t$, $r_t$ are for $k=1,\cdots,N$, square integrable processes which will be chosen at each single step of the induction procedure. Also $\zeta$ is a $L^2(\Omega, {\cal F}_T)$ random variable. Observe that if $\lambda=0$,   the system (\ref{fo:Peng_Wu_system}) is a particular case of the system in Lemma 2.5 in \cite{Peng1999} for which existence and uniqueness is established, and it becomes the system (\ref{fo:FABSDE_centered}) when setting $\lambda=1$, $\zeta^{i,i}=0$, $\phi^i_t=0$, $\psi^{i,i,k}_t=0$, $r^{i,i}_t=0$, $i=1,\cdots,N$ and  $k=1,\cdots,N$, for $0\leq t \leq T$. We only give the proof of  existence and uniqueness for the solution  of the system \eqref{fo:FABSDE_centered} in the case of $c=0$. The same arguments can be used to treat the case $c>0$. 

The proof relies on the following technical result which we prove in the appendix.

\begin{lemma}\label{thm-2}
If there exists $\lambda_0\in[0,1)$ such that for any $\zeta$ and $\phi_t$,  $r_t$, $\psi^{k}_t$, $k=1,\cdots,N$ for $0\leq t \leq T$ the system (\ref{fo:Peng_Wu_system}) admits a unique solution for $\lambda=\lambda_0$, then there exists $\kappa_0>0$, such that for all $\kappa\in[0,\kappa_0)$, (\ref{fo:Peng_Wu_system}) admits a unique solution for any $\lambda\in[\lambda_0,\lambda_0+\kappa)$. 
\end{lemma}

Taking for granted the result of this lemma, we can prove existence and uniqueness for \eqref{fo:Peng_Wu_system}. Indeed, for $\lambda=0$, the result is known. Using Lemma \ref{thm-2}, there exists $\kappa_0>0$ such that  \eqref{fo:Peng_Wu_system} admits a unique solution for $\lambda=0+\kappa$ where $\kappa\in[0,\kappa_0)$. Repeating the inductive argument $n$ times for $1\leq n\kappa_0<1+\kappa_0$ gives the result for $\lambda=1$ and, therefore, the existence of the unique solution for \eqref{fo:FABSDE_centered}. 
Since $X_t^{i,c}=X^i_t-\overline{X}_t$, $Y_t^{i,i,c}=Y^{i,i}_t$ and $Z_t^{i,i,k,c}={Z}_t^{i,i,k}$, and 
$\overline{X}_t$ is given by (\ref{eq:Xbar}), we obtain a unique solution $(X^i_t, Y^{i,i}_t,{Z}^{i,i,k}_t)$ to the system \eqref{fo:FABSDE_general}.
\end{proof}

\section{Hamilton-Jacobi-Bellman (HJB) Approach}\label{sec:HJB}

In this section, we return to the particular case $\theta=\delta_0-\delta_\tau$ of the drift given by the delayed control $\alpha_t-\alpha_{t-\tau}$.
The HJB approach for delayed systems has been applied by \cite{vinter1981infinite} to a deterministic linear quadratic control problem. Later, \cite{F.Gozzi2004} followed a similar
approach for stochastic control problems. Here, we generalize the approach \cite{F.Gozzi2004} based on an infinite dimensional 
representation and functional derivatives. We extend this approach to 
our stochastic game model with delay in order to identify a closed-loop Nash equilibrium.  

Note that two specific features of our discussion require additional work for our argument to be fully rigorous at the mathematical level.  First, the delayed control in the state equation appears as a mass at time $t-\tau$ and a smoothing argument as in \cite{F.Gozzi2015} is needed. Second, we are using functional derivatives and proper function spaces should be introduced for our computations to be fully justified. However, since most of the functions we manipulate are linear or quadratic, we refrain from giving the details. In that sense, and for these two reasons, what follows is merely heuristic. A rigorous proof of the fact that the equilibrium identified in this section is actually a Nash equilibrium will be given in Section \ref{caratheodory}.

\subsection{Infinite Dimensional Representation}
Let $\mathbb{H}^{N}$ be the Hilbert space defined by
\[
 \mathbb{H}^N=\mathbb{R}^{N} \times L^{2}([-\tau,0];\mathbb{R}^{N}),
\]
 with the inner product 
\[
\langle z,\tilde{z} \rangle=z_{0}\tilde{z}_{0}+\int_{-\tau}^{0} z_{1}(\xi)\tilde{z}_{1}(\xi)\,d\xi,
\]
where $z,\tilde{z} \in \mathbb{H}^N$, and $z_{0}$ and $z_{1}(.)$ correspond respectively to the $\mathbb{R}^{N}$-valued and $L^{2}([-\tau,0];\mathbb{R}^{N})$-valued 
components.\par
By reformulating the system of coupled diffusions (\ref{Xi-1}) in the Hilbert space $\mathbb{H}^N$, the system of coupled Abstract Stochastic Differential Equations 
(ASDE) for $Z=(Z^1,\cdots,Z^N) \in \mathbb{H}^N$ appears as
\ba
\label{asde}
d Z_t &=&\left(A Z_t+B \alpha_t \right)dt + G d{W}_t,\quad 0\leq t\leq T,  \\
Z_0 &=& (\xi,0)\in \mathbb{H}^N. \nonumber
\ea
where ${W}_t=(W^1_t,\cdots,W^N_t)$ is a standard $N$-dimensional Brownian motion and $\xi=(\xi^1,\cdots,\xi^N)$.

Here $Z_t = (Z_{0,t},Z_{1,t,r})$, $r \in [-\tau,0]$ corresponds to $(X_t,\alpha_{t-\tau-r})$ in 
the system of diffusions (\ref{Xi-1}). In other words, for each time $t$, in order to find 
the dynamics of the states $X_t$, it is necessary to have $X_t$ itself, and the past of the control $\alpha_{t-\tau-r}$, $r \in [-\tau,0]$.


The operator $A:\it{D}(A) \subset \mathbb{H}^N \rightarrow \mathbb{H}^N$ is defined as
\[
 A:(z_0,z_1(r)) \rightarrow (z_1(0),-\frac{dz_1(r)}{d r}) \hspace{3mm} a.e., \hspace{2mm} r \in [-\tau,0],
\]
and its domain  is
\[
\it{D}(A)=\{(z_0,z_1(.)) \in \mathbb{H}^N: z_1(.) \in W^{1,2}([-\tau,0];\mathbb{R}^{N}),z_1(-\tau)=0\}.
\]
The adjoint operator of $A$ is $A^{\ast}:\it{D}(A^{\ast}) \subset \mathbb{H}^N \rightarrow \mathbb{H}^N$ and is defined by
\[
 A^{\ast}:(z_0,z_1(r)) \rightarrow (0,\frac{dz_1(r)}{d r}) \hspace{3mm} a.e., \hspace{2mm} r \in [-\tau,0],
\]
with domain 
\[
 \it{D}(A^{\ast})=\{(w_0,w_1(.)) \in \mathbb{H}^N: w_1(.) \in W^{1,2}([-\tau,0]);\mathbb{R}^{N}),w_0=w_1(0)\}.
\]
The operator $B:\mathbb{R}^N \rightarrow \mathbb{H}^N$ is defined by
\[
 B:u \rightarrow (u,-\delta_{-\tau}(r)u), \hspace{2mm} r \in [-\tau,0],
\]
where $\delta_{-\tau}(.)$ is the Dirac measure at $-\tau$.\\
\begin{remark}
 Note that in \cite{F.Gozzi2004}, the case of pointwise delay is not considered as the above operator $B$ becomes unbounded because of the dirac measure. Here, we still use 
the unbounded operator $B$ (in a heuristic sense!) and for a rigorous treatment, we refer
 to \cite{F.Gozzi2015} where they use partial
smoothing to accommodate the case of pointwise delay. 
\end{remark}

Finally, the operator $G: \mathbb{R}^N \rightarrow \mathbb{H}^N$ is defined by
\[
 G: z_0 \rightarrow (\sigma z_0,0).
\]
\begin{remark}
 Let $Z_t$ be a weak solution of the system of coupled ASDEs (\ref{asde}) and $X_t$ be a
 continuous solution of the system of diffusions (\ref{Xi-1}), then, with a similar line of reasoning as in Proposition 2  in \cite{F.Gozzi2004}, it can be proved 
that $X_t = Z_{0,t}$, a.s. for all $t \in [0,T]$.
\end{remark}
\subsection{System of Coupled HJB Equations}

In order to use the dynamic programming principle for stochastic games  (we refer to \cite{CarmonaSIAM2016}) in search of closed-loop Nash equilibrium, the
 initial time is varied. At time $t \in [0,T]$,
 given initial state $Z_{t}=z$  (whose second component is the past of the control), bank $i$ chooses the control $\alpha^i$ to minimize its objective function $J^{i}(t,z,\alpha)$.
\ba
J^{i}(t,z,\alpha)=\EE \bigg\{\int_{t}^{T}f_i(Z_{0,s},\alpha^i_s)dt+g_i(Z_{0,T}) \mid Z_{t}=z \bigg\},
\ea
In equilibrium, that is all other banks $j\neq i$ have optimized their objective function, bank $i$'s value function $V^i(t,z)$ is  
\ba
 V^i(t,z)=\inf_{\alpha^i} J^i(t,z,\alpha).
\ea
The set of value functions $V^i(t,z)$, $i=1,\cdots,N$ is a solution (in a suitable sense) of the following system of coupled HJB equations:
\ba
\label{HJB}
 && \partial_t V^i +\frac{1}{2} Tr(Q \partial_{z z} V^i)+ \langle A z,\partial_{z} V^i \rangle + H^i_{0}(\partial_{z}V^i) = 0 ,  \\
&& V^i(T)= g_i ,  \nonumber
\ea
where $Q=G\mathbin{*} G$, and the Hamiltonian function $H^i_0(p^i):\mathbb{H}^N \rightarrow \mathbb{R}$ is defined by
\ba
 H^i_0(p^i) =\inf_{\alpha^i} [ \langle B \alpha,p^i \rangle+ f_i(z_{0},\alpha^i)].
\ea
Here, $p^i \in \mathbb{H}^N$ and can be written as $p^i=(p^{i,1},\cdots,p^{i,N})$ where $p^{i,k} \in \mathbb{H}^1$, $k=1,\cdots,N$.
 Given that $f_i(z_0,\alpha^i)$ is convex in $(z_0,\alpha^i)$, 
\ba
\label{zbest}
\hat{\alpha}^{i}=-\langle B,p^{i,i} \rangle - q (z^i_0 -\bar{z}_0).
\ea
Therefore,
\ba
 H_0^i(p) &&=  \langle B \hat{\alpha},p^{i} \rangle+ f_{i}(z_0,\hat{\alpha}^{i}), \nonumber \\
	  &&= \sum\limits_{k=1}^{N} \langle B ,p^{i,k} \rangle \left(-\langle B,p^{k,k} \rangle - q (z^k_0 -\bar{z}_0)\right)   \nonumber \\
           &&+ \frac{1}{2}{\langle B,p^{i,i} \rangle}^2+\frac{1}{2}(\epsilon-q^2)(\bar{z}_0 -z^i_0)^2.
\ea
We then make the ansatz
\ba
\label{ansatz}
 V^i(t,z) &=& E_0(t) (\bar{z}_0-z_0^i)^2-2 (\bar{z}_0-z_0^i) \int\limits_{-\tau}^{0}E_1(t, -\tau-s)(\bar{z}_{1,s}-z^i_{1,s}) d s \nonumber  \\
       &+& \int\limits_{-\tau}^{0} \int\limits_{-\tau}^{0} E_2(t, -\tau-s, -\tau-r) 
(\bar{z}_{1,s}-z^i_{1,s}) (\bar{z}_{1,r}-z^i_{1,r}) d s d r + E_3(t),  \nonumber \\
\label{tets4}
\ea
where $E_
0(t)$, $E_
1(t,s)$, $E_
2(t,s,r)$ and $E_
3(t)$ are some deterministic functions to be determined. It is assumed that $E_
2(t,s,r)=E_
2(t,r,s)$. 
%

\begin{remark}
 Note that the ansatz (\ref{ansatz}) depends on $z \in \mathbb{H}^N$ whose second component is the past of all banks' controls $\alpha$. In other words, the 
value function $V^i(t,z)$ is an explicit function of the past of all banks' controls $\alpha_{t-\tau-r}$, $r \in [-\tau,0]$.
\end{remark}


The derivatives of the ansatz (\ref{ansatz}) are as follows
\ba
  \partial_t V^i &=& \frac{d E_
0(t)}{d t}(\bar{z}_0-z_0^i)^2-
 2 (\bar{z}_0-z_0^i) \int\limits_{-\tau}^{0}\frac{\partial E_
1(t, -\tau-s)}{\partial t}(\bar{z}_{1,s}-z^i_{1,s}) d s  \nonumber \\
&&+ \int\limits_{-\tau}^{0} \int\limits_{-\tau}^{0} \frac{\partial E_
2(t, -\tau-s, -\tau-r)}{\partial t} 
(\bar{z}_{1,s}-z^i_{1,s}) (\bar{z}_{1,r}-z^i_{1,r}) d s d r  
+\frac{d E_
3(t)}{d t} ,   \nonumber \\
\\
\partial_{z^j} V^i &=& \begin{bmatrix}
        2 E_
0(t)(\bar{z}_0-z_0^i)-2\int\limits_{-\tau}^{0}E_
1(t,-\tau- s)(\bar{z}_{1,s}-z^i_{1,s}) d s \\
	-2 (\bar{z}_0-z_0^i) E_
1(t, s)+2\int\limits_{-\tau}^{0} E_
2(t,-\tau-s,-\tau-r) (\bar{z}_{1,r}-z^i_{1,r})d r
       \end{bmatrix}\left(\frac{1}{N}-\delta_{i,j}\right), \nonumber \\
\\
\partial_{z^j z^k} V^i &=& \begin{bmatrix}
        2 E_
0(t) & -2 E_
1(t,-\tau-s) \\
	-2 E_
1(t,-\tau-s) & 2E_
2(t,-\tau-s,-\tau-r)
       \end{bmatrix}\left(\frac{1}{N}-\delta_{i,j}\right)\left(\frac{1}{N}-\delta_{i,k}\right).
\ea
By plugging the ansatz (\ref{ansatz}) into the HJB equation (\ref{HJB}), and collecting all the corresponding terms, the following set of equations are derived
for $t \in [0,T]$ and $s, r \in [-\tau,0]$.\par
The equation corresponding to the \emph{constant} terms is
\ba
\label{eq1-1}
 \frac{d E_
3(t)}{d t}+ (1-\frac{1}{N})\sigma^2 E_
0(t) =0,
\ea

The equation corresponding to the $(\bar{z}_0 -z^i_0)^2$ terms is
\ba
\label{eq1-2}
 & \frac{d E_
0(t)}{d t} + \frac{\epsilon}{2}=2(1-\frac{1}{N^2})(E_
1(t,0)+E_
0(t))^2+2q(E_
1(t,0)+E_
0(t)) +\frac{q^2}{2}.
\ea
The equation corresponding to the $(\bar{z}_0 -z^i_0) (\bar{z}_1-z^i_1)$ terms is
\ba
\label{eq1-3}
& \frac{\partial E_
1(t, s)}{\partial t}-\frac{\partial E_
1(t, s)}{\partial s}=2(1-\frac{1}{N^2})
\left(E_
1(t,0)+E_
0(t)+\frac{q}{2(1-\frac{1}{N^2})}\right)\left(E_
2(t,s,0)+E_
1(t,s)\right). 
\ea

The equation corresponding to the $(\bar{z}_1-z^i_1)(\bar{z}_1-z^i_1)$ terms is
\ba
\label{eq1-4}
&  \frac{\partial E_
2(t, s, r)}{\partial t}-\frac{\partial E_
2(t, s, r)}{\partial s} -\frac{\partial E_
2(t, s, r)}{\partial r}= \nonumber  \\
& 2(1-\frac{1}{N^2})\left(E_
2(t,s,0)+E_
1(t,s)\right)\left(E_
2(t,r,0)+E_
1(t,r)\right).
\ea

The boundary conditions are
\ba\label{bc}
&& E_
0(T)=\frac{c}{2}, \nonumber \\
&& E_
1(T,s)=0, \nonumber \\ 
&& E_
2(T,s,r)=0,  \nonumber \\  
&& E_
2(t,s,r)=E_
2(t,r,s),   \\  
&& E_
1(t,-\tau)= -E_
0(t), \hspace{3mm} \forall t \in [0,T),  \nonumber \\
&& E_
2(t,s,-\tau)= -E_
1(t,s), \hspace{3mm} \forall t \in [0,T),  \nonumber \\   
&& E_
3(T)=0. \nonumber
\ea
Note that with these boundary conditions (at $t=T$), we have $V^i(T,z)=g_i(z_0)=\frac{c}{2}(\bar{z}_0-z_0^i)^2$, as desired.
\begin{remark}
\label{remarkE}
 The set of equations (\ref{eq1-1}--\ref{eq1-4}) on the domain $t \in [0,T]$, $s, r \in [-\tau,0]$, and with boundary conditions (\ref{bc})
admits a unique solution. This can be shown by following the
steps of the proof of Theorem 6 in \cite{alekal1971quadratic} and using a fixed point argument (see also \cite{L.Chen2011}).
\end{remark}
If all the other banks choose their optimal controls, then the bank $i$'s optimal strategy $\hat{\alpha}^{i}$,  $i=1,\cdots,N$ follows 

\ba
 \hat{\alpha}^{i}_t &&= -\langle B, \partial_{z^i} V^i \rangle - q (z^i_{0} -\bar{z}_{0}),  \nonumber \\
&& =2\left(1-\frac{1}{N}\right) \left[\left(E_
1(t,0)+E_
0(t)+\frac{q}{2\left(1-\frac{1}{N}\right)}\right)(\bar{z}_{0}-z^i_{0}) \right. \nonumber \\
&&\left.  - \int_{-\tau}^{0}\left(E_
2(t,-\tau-s,0)+E_
1(t,-\tau-s)\right)(\bar{z}_{1,s}-z^i_{1,s}) d s \right] . 
\ea
In terms of the original system of coupled diffusions (\ref{Xi-1}), the closed-loop Nash equilibrium corresponds to 
\ba\label{Nash}
 \hat{\alpha}^{i}_t &&=  2\left(1-\frac{1}{N}\right)\left[\left(E_
1(t,0)+E_
0(t)+\frac{q}{2\left(1-\frac{1}{N}\right)}\right)(\bar{X}_{t}-X^i_{t}) \right. \nonumber \\
&&\left.  + \int_{t-\tau}^{t}\left[E_
2(t,s-t,0)+E_
1(t,s-t)\right](\bar{\hat{\alpha}}_{s}-\hat{\alpha}^i_{s}) d s\right]  ,\quad 
 i=1,\cdots, N.   \nonumber \\
\ea

\begin{remark}
\label{re:clearing2}
As pointed out in Remark \ref{re:clearing} of Section \ref{Systemic-Risk-SDDE}, in the present situation we still have $\sum_{i=1}^N\hat\alpha^i_t=0$ and therefore, in this  equilibrium, the central bank serves as
a clearing house (see also the discussion of Section \ref{FI}).  
\end{remark}

\section{A Verification Theorem} \label{caratheodory}

 In this section, we provide a verification theorem establishing that the strategies given by (\ref{Nash}) correspond to a Nash equilibrium. Our solution is only \emph{almost explicit} because the equilibrium strategies are given by the solution of a system of integral equations.
This approach has been used by \cite{alekal1971quadratic} to find the optimal control in a deterministic delayed linear quadratic control problem.  Recently, \cite{L.Chen2011} and
 \cite{huang2012forward} have 
applied this approach to delayed linear quadratic stochastic control problems. In this section, we generalize it to delayed linear -quadratic stochastic games differential games. 

We recall that 
at time $t \in [0,T]$, given $x=(x^1,\cdots,x^N)$, which should be viewed as the state of the $N$ banks at time $t$, and an $A$-valued function $\alpha$ on $[0,\tau)$,
which should be viewed as their collective controls over the time interval $[t-\tau,t)$, bank $i$ chooses the strategy $\alpha^i$ to minimize its objective function
\ba
J^{i}(t,x,\alpha,\alpha^{t})=\EE\left\{\int_{t}^{T}f_i(X_s,\alpha^i_s)ds+g_i(X_T) \mid X_t=x, \alpha_{[t)}=\alpha\right\}.
\ea
Here $\alpha_{[t)}$ is defined as the restriction of the path $s\mapsto \alpha_s$ to the interval $[t-\tau,t)$
and $\alpha^t$ is an admissible control strategy for the $N$ banks over the time interval $[t,T]$. We denote by $\A^t$ this set of admissible strategies.

In the search for Nash equilibria, for each bank $i$, we assume that the banks $j\ne i$ chose their strategies $\alpha^{-i,t}$ for the \emph{future} $[t,T]$, in which case, bank $i$'s should choose a strategy $\alpha^{i,t}\in\A^{i,t}$ in order to try to minimize its objective function $J^i(t,x,\alpha,(\alpha^{i,t},\alpha^{-i,t}))$. As a result we define the value function $V^i(t,x,\alpha,\alpha^{-i,t})$ of bank $i$ by:  
\ba
V^i(t,x,\alpha,\alpha^{-i,t})=\inf_{\alpha^{i,t}\in\A^{i,t}} J^i(t,x,\alpha,(\alpha^{i,t},\alpha^{-i,t})).
\ea
Because of the linear nature of the dynamics of the states, together with the quadratic nature of the costs, we expect that in equilibrium, the functions $J^i$ and $V^i$ to be quadratic functions of the state $x$ and the past $\alpha$ of the control. This is consistent with the choices we made in the previous section. Accordingly, we write the functions $V^i$ as
\ba
V^i(t,x,\alpha) &=& E_0(t) (\bar{x}-x^i)^2+2 (\bar{x}-x^i) \int\limits_{t-\tau}^{t}E_1(t, s-t)(\bar{\alpha}_s-\alpha^i_s) d s \nonumber  \\
       &+& \int\limits_{t-\tau}^{t} \int\limits_{t-\tau}^{t} E_2(t, s-t, r-t) 
(\bar{\alpha}_s-\alpha^i_s) (\bar{\alpha}_r-\alpha^i_r) d s d r + E_3(t),
\label{tets3}
\ea
where the deterministic functions $E_i\,(i=0,\cdots,3)$, are the solutions of the system (\ref{eq1-1}--\ref{eq1-4}) with the boundary conditions (\ref{bc}). We dropped the dependence of $V^i$ upon its fourth parameter $\alpha^{-i,t}$ because the right hand side of \eqref{tets3} does not depend upon $\alpha^{-i,t}$.

\vskip 6pt
The main result of this section is Proposition \ref{prop1} below which says that any solution of the system \eqref{Nash} of integral equations provides a Nash equilibrium. For that reason, we first prove existence and uniqueness of solutions of these integral equations when they are recast as a fixed point problem in 
classical spaces of adapted processes. This is done in Lemma \ref{le:fixed_point} below. We simplify the notation and we rewrite equation \eqref{Nash} for the purpose of the proof of the lemma. We set:
$$
\varphi(t)=2\left(1-\frac{1}{N}\right)\left(E_
1(t,0)+E_
0(t)+\frac{q}{2\left(1-\frac{1}{N}\right)}\right)
$$
and 
$$
\bar\psi(t,s) = [E_2(t,s-t,0)+E_1(t,s-t)]{\bf 1}_{[t-\tau,t]}(s)
$$
so that equation \eqref{Nash} can be rewritten as:
\begin{equation}
\label{Nash2}
\begin{split}
\hat{\alpha}^{i}_t &=  \varphi(t)(\bar{X}_{t}-X^i_{t}) + \int_{0}^{t}\bar\psi(t,s)(\bar{\hat{\alpha}}_{s}-\hat{\alpha}^i_{s}) d s\\
&=  \varphi(t)\left((\bar\xi-\xi^i) - \int_{0}^{t}[(\bar{\hat\alpha}_s-\hat\alpha^i_s)-(\bar{\hat\alpha}_{s-\tau}-\hat\alpha^i_{s-\tau})]ds +\sigma[\bar W_t-W^i_t]\right)
+ \int_{0}^{t}\bar\psi(t,s)(\bar{\hat{\alpha}}_{s}-\hat{\alpha}^i_{s}) d s.
\end{split}
\end{equation}
Summing these equations for $i=1,\cdots,N$, we see that any solution should necessarily satisfy $\sum_{1\le i\le N}\hat\alpha^i=0$, so that if we look for a solution of the system \eqref{Nash}, we might as well restrict our search to processes satisfying $\bar{\hat\alpha}_t=0$ for all $t\in[0,T]$.

\vskip 2pt
So we denote by $\RR^N_0$ the set of elements $x=(x^1,\cdots,x^N)$ of $\RR^N$ satisfying $\sum_{1\le i\le N}x^i=0$, and by $\cH^{2,N}_0$ the space of $\RR^N_0$-valued adapted processes $a=(a_t)_{0\le t\le T}$ satisfying
$$
 \|a\|^2_0:=\EE\bigg[\int_0^T|a_t|^2dt\bigg]<\infty.
$$
Clearly, $\cH^{2,N}_0$ is a real separable Hilbert space for the scalar product derived from the norm $\|\,\cdot\,\|_0$ by polarization.
For $a\in\cH^{2,N}_0$ we define the $\RR^N_0$-valued process $\Psi(a)$ by:
\begin{equation}
\label{fo:Psi}
\Psi(a)^i_t= \varphi(t)(\bar\xi-\xi^i) +\sigma\varphi(t)[\bar W_t-W^i_t] +\int_0^t\psi(t,s)a^i_sds,
\qquad 0\le t\le T,\; i=1,\cdots,N.
\end{equation}
where the function $\psi$ is defined by $\psi(t,s)=1-{\bf 1}_{[0,0\vee(t-\tau)]}(s)-\bar\psi(t,s)$. We shall use the fact that the functions $\varphi$ and $\psi$ are bounded.

Given the above set-up, existence and uniqueness of a solution to \eqref{Nash} is given by the following lemma whose proof mimics the standard proofs of existence and uniqueness of solutions of stochastic differential equations.

\begin{lemma}
\label{le:fixed_point}
The map $\Psi$ defined by \eqref{fo:Psi} has a unique fixed point in $\cH^{2,N}_0$.
\end{lemma}

%



\vskip 6pt
We now prove existence of Nash equilibria for the system.
\begin{prop}
\label{prop1}
The strategies $(\hat\alpha^i_t)_{0\le t\le T,\,i=1,\cdots,N}$ given by the solution of the system of integral equations (\ref{Nash})
form a Nash equilibrium, and the corresponding value functions are given by (\ref{tets3}). 
\end{prop}
In other words, we prove that
$$
V^i(0,\xi^i,\alpha_{[0)}) \leq J^i(0,\xi^i, \alpha_{[0)},(\alpha^i,\hat\alpha^{-i})),
$$
for any $\alpha^i$, and choosing $\alpha^i=\hat\alpha^i$ gives: 
$$
V^i(0,\xi^i,\alpha_{[0)}) = J^i(0,\xi^i, \alpha_{[0)},(\hat\alpha^i,\hat\alpha^{-i})).
$$
Notice that the equilibrium strategies which we identified are in feedback form in the sense that each $\hat\alpha^i_t$ is a deterministic function of the trajectory $X_{[0,t]}$ of the past of the state. Notice also that there is absolutely nothing special with the time $t=0$ and the initial condition $X_0=\xi, \alpha_{[0)}=0$. Indeed for any $t\in [0,T]$ and $\RR^N$-valued square integrable random variable $\zeta$, the same proof can be used to construct a Nash equilibrium for the game over the interval $[t,T]$ and any initial condition $(X_t=\zeta,\alpha_{[t)})$.


\begin{proof}
We fix an arbitrary $i\in\{1,\cdots,N\}$, an admissible control $\alpha^i\in\A^{-i}$ for player $i$, and we assume that the state process $(X_t)_{0\le t\le T}$ for the $N$ banks is controlled by $(\alpha^i_t,\hat\alpha^i_t)_{0\le t\le T}$ where $(\hat\alpha^k_t)_{0\le t\le T,\,k=1,\cdots,N}$ solves the system of integral equations (\ref{Nash}). Next, we apply It\^o's formula to $V^i(t,X_t,\alpha_{[t)})$ where the function $V^i$ is defined by (\ref{tets3}). We obtain

\ba
&& d V^i(t,X_t,\alpha_{[t)}) = \nonumber \\
&& \Bigg\{   \frac{d E_0(t)}{d t}(\bar{X}_t-X^i_t)^2 + 2 E_0(t)(\bar{X}_t-X^i_t)\left(\bar{\alpha}_t-\alpha^i_t-(\bar{\alpha}_{t-\tau}-\alpha^i_{t-\tau})\right)
    \nonumber \\
&&+\sum\limits_{j=1}^{N}\sigma^2 E_0(t)(\frac{1}{N}-\delta_{i,j})^2 + 2\left(\bar{\alpha}_t-\alpha^i_t-(\bar{\alpha}_{t-\tau}-\alpha^i_{t-\tau})\right) \int\limits_{t-\tau}^{t}E_1(t, s-t)(\bar{\alpha}_s-\alpha^i_s) d s    \nonumber \\
&&+ 2 (\bar{X}_t-X^i_t) \int\limits_{t-\tau}^{t}\left[\frac{\partial E_1(t, s-t)}{\partial t}-\frac{\partial E_1(t, s-t)}{\partial s}\right](\bar{\alpha}_s-\alpha^i_s) d s \nonumber \\
&&+ 2 (\bar{X}_t-X^i_t) E_1(t,0)(\bar{\alpha}_t-\alpha^i_t)- 2 (\bar{X}_t-X^i_t) E_1(t,-\tau)  (\bar{\alpha}_{t-\tau}-\alpha^i_{t-\tau})  \nonumber \\
&&+  \int\limits_{t-\tau}^{t} \int\limits_{t-\tau}^{t} \bigg[\frac{\partial E_2(t, s-t, r-t)}{\partial t}-\frac{\partial E_2(t, s-t, r-t)}{\partial s}  \nonumber \\
&&- \frac{\partial E_2(t, s-t, r-t)}{\partial r}\bigg] (\bar{\alpha}_s-\alpha^i_s) (\bar{\alpha}_r-\alpha^i_r) d s d r   \nonumber \\
&&+ (\bar{\alpha}_t-\alpha^i_t) \left(\int\limits_{t-\tau}^{t} E_2(t, s-t,0)(\bar{\alpha}_s-\alpha^i_s) d s + \int\limits_{t-\tau}^{t} E_2(t, 0,r-t)(\bar{\alpha}_r-\alpha^i_r) d r \right)  \nonumber \\
&&- (\bar{\alpha}_{t-\tau}-\alpha^i_{t-\tau}) \left(\int\limits_{t-\tau}^{t} E_2(t, s-t,-\tau)(\bar{\alpha}_s-\alpha^i_s) d s + 
\int\limits_{t-\tau}^{t} E_2(t,-\tau,r-t)(\bar{\alpha}_r-\alpha^i_r) d r \right)  \nonumber \\
&&+ \frac{d E_3(t)}{d t} \Bigg\} dt   \nonumber \\
&&+ \sum\limits_{j=1}^{N}\Bigg\{
+2 E_0(t)(\bar{X}_t-X^i_t)(\frac{1}{N}-\delta_{i,j}) + 2 (\frac{1}{N}-\delta_{i,j})\int\limits_{t-\tau}^{t} E_1(t, s-t)(\bar{\alpha}_s-\alpha^i_s) d s \Bigg\} \sigma  d W^j_t.\nonumber\\
\ea
Then, integrating between $0$ and $T$, using $V^i(T,X_T)=g_i(X_T)$ (ensured by the boundary conditions at $t=T$ for $E_k, \,k=0,1,2,3$), taking expectation,  using the differential equations (\ref{eq1-1}-\ref{eq1-4}), using the short notation $A_1=1-\frac{1}{N},\, A_2=1-\frac{1}{N^2}$,
and adding $\mathbb{E} \int\limits_{0}^{T}f_i(X_s,\alpha^i_s) d t $ on both sides, one obtains:
\ba
\label{duadequ}
&& - V^i(0,\xi^i,\alpha_{[0)}) + \mathbb{E}(g_i(X_T))+   \mathbb{E} \int\limits_{0}^{T}f_i(X_s,\alpha^i_s) d t =   - V^i(0,\xi^i,\alpha_{[0)}) +J^i(0,\xi^i, \alpha_{[0)},\alpha) =  \nonumber \\
&& \mathbb{E} \int\limits_{0}^{T}\Bigg\{ 
\left[-\frac{\epsilon}{2}+2A_2(E_1(t,0)+E_0(t))^2+2q(E_1(t,0)+E_0(t))+\frac{q^2}{2}\right](\bar{X}_t-X^i_t)^2\nonumber\\
&&+2E_0(t)(\bar{X}_t-X^i_t)\left((\bar{\alpha}_t-\alpha^i_t)-(\bar{\alpha}_{t-\tau}-\alpha^i_{t-\tau})\right)
+\sigma^2E_0(t)\sum_{j=1}^N(\frac{1}{N}-\delta_{i,j})^2\nonumber\\
&&+ 2\left(\bar{\alpha}_t-\alpha^i_t-(\bar{\alpha}_{t-\tau}-\alpha^i_{t-\tau})\right) \int\limits_{t-\tau}^{t}E_1(t, s-t)(\bar{\alpha}_s-\alpha^i_s) d s    \nonumber \\
&&+ 2 (\bar{X}_t-X^i_t) \int\limits_{t-\tau}^{t}\left[2A_2
\left(E_1(t,0)+E_0(t)+\frac{q}{2A_2}\right)\left(E_2(t,s-t,0)+E_1(t,s-t)\right)\right](\bar{\alpha}_s-\alpha^i_s) d s \nonumber \\
&&+ 2 (\bar{X}_t-X^i_t) E_1(t,0)(\bar{\alpha}_t-\alpha^i_t)- 2 (\bar{X}_t-X^i_t) E_1(t,-\tau)  (\bar{\alpha}_{t-\tau}-\alpha^i_{t-\tau})  \nonumber \\
&&\hskip -1cm +  \int\limits_{t-\tau}^{t} \int\limits_{t-\tau}^{t} \bigg[2A_2\left(E_2(t,s-t,0)+E_1(t,s-t) \right) \left(E_2(t,r-t,0)+E_1(t,r-t) \right)\bigg] (\bar{\alpha}_s-\alpha^i_s) (\bar{\alpha}_r-\alpha^i_r) d s d r   \nonumber \\
&&+ (\bar{\alpha}_t-\alpha^i_t) \left(\int\limits_{t-\tau}^{t} E_2(t, s-t,0)(\bar{\alpha}_s-\alpha^i_s) d s + \int\limits_{t-\tau}^{t} E_2(t, 0,r-t)(\bar{\alpha}_r-\alpha^i_r) d r \right)  \nonumber \\
&&- (\bar{\alpha}_{t-\tau}-\alpha^i_{t-\tau}) \left(\int\limits_{t-\tau}^{t} E_2(t, s-t,-\tau)(\bar{\alpha}_s-\alpha^i_s) d s + 
\int\limits_{t-\tau}^{t} E_2(t,-\tau,r-t)(\bar{\alpha}_r-\alpha^i_r) d r \right)  \nonumber \\
&&-A_1\sigma^2E_0(t) +\frac{1}{2}(\alpha_t^i)^2-q\alpha_t^i(\bar{X}_t-X^i_t)+\frac{\epsilon}{2}(\bar{X}_t-X^i_t)^2
\Bigg\} dt .
\ea
Observe that the terms in $\epsilon$ cancel, the terms in $\sigma^2$ cancel, and the terms involving delayed controls cancel using symmetries and boundary conditions (\ref{bc}) for the functions $E_k$'s. 

Next, motivated by (\ref{Nash}), we rearrange the terms  left in (\ref{duadequ}) so that the square of 
\[
 {\alpha}^i_t-2A_1\left[\left(E_1(t,0)+E_0(t)+\frac{q}{2A_1}\right)(\bar{X}_{t}-X^i_{t}) +\int_{t-\tau}^{t}\left[E_2(t,s-t,0)+E_1(t,s-t)\right](\bar{\hat{\alpha}}_{s}-\hat{\alpha}^i_{s}) d s \right]
 \]
appears first. We obtain
\ba
\label{duadequ2}
&&  - V^i(0,\xi^i,\alpha_{[0)})+J^i(0,\xi^i, \alpha_{[0)},\alpha) =  \nonumber \\
&&  \mathbb{E} \int\limits_{0}^{T}\Bigg\{ 
\frac{1}{2}\left({\alpha}^i_t-2A_1\left[\left(E_1(t,0)+E_0(t)+\frac{q}{2A_1}\right)(\bar{X}_{t}-X^i_{t}) \right.\right.
\nonumber \\
&&\hskip 4cm \left.\left. +\int_{t-\tau}^{t}\left[E_2(t,s-t,0)+E_1(t,s-t)\right](\bar{\hat{\alpha}}_{s}-\hat{\alpha}^i_{s}) d s \right]\right)^2\nonumber\\
&&+(\bar{X}_t-X^i_t)^2\left[-2[A_1(E_1(t,0)+E_0(t)+\frac{q}{2}]^2+2A_2(E_1(t,0)+E_0(t))^2+2q(E_1(t,0)+E_0(t))+\frac{q^2}{2}\right]\nonumber\\
&&+(\bar{X}_t-X^i_t)\left[2\alpha_t^i[A_1(E_1(t,0)+E_0(t)]+2(E_1(t,0)+E_0(t))(\bar\alpha_t-\alpha_t^i)\right]\nonumber\\
&&+(\bar{X}_t-X^i_t)\left(\int_{t-\tau}^t (E_2(t,s-t,0)+E_1(t,s-t)(\bar\alpha_s-\alpha_s^i)ds\right)\left[
-4A_1\left(A_1(E_1(t,0)+E_0(t)+\frac{q}{2}\right)\right.\nonumber\\
&&\hskip 9cm\left.+4A_2\left(E_1(t,0)+E_0(t)+\frac{q}{2A_2}\right)\right]\nonumber\\
&&+\left(\int_{t-\tau}^t (E_2(t,s-t,0)+E_1(t,s-t)(\bar\alpha_s-\alpha_s^i)ds\right)\left[
2A_1\alpha_t^i+2(\bar\alpha_t-\alpha_t^i)\right]\nonumber\\
&&+\left(\int_{t-\tau}^t (E_2(t,s-t,0)+E_1(t,s-t)(\bar\alpha_s-\alpha_s^i)ds\right)^2\left[
-2A_1^2+2A_2\right]
\Bigg\} dt .
\ea
Using $A_2=A_1^2+\frac{2}{N}A_1$ and the relation $\bar\alpha_t-\alpha_t^i=\frac{1}{N}\sum\limits_{j\neq i}\alpha^j_t-A_1\alpha_t^i$, we simplify (\ref{duadequ2}) to obtain:
\ba
\label{duadequ3}
&&  - V^i(0,\xi^i,\alpha_{[0)}) +J^i(0,\xi^i, \alpha_{[0)},\alpha) =  \nonumber \\
&& \mathbb{E} \int\limits_{0}^{T}\Bigg\{ 
\frac{1}{2}\left({\alpha}^i_t-2A_1\left[\left(E_1(t,0)+E_0(t)+\frac{q}{2A_1}\right)(\bar{X}_{t}-X^i_{t}) 
\right.\right.
\nonumber \\
&&\hskip 4cm \left.\left.
+\int_{t-\tau}^{t}\left[E_2(t,s-t,0)+E_1(t,s-t)\right](\bar{\hat{\alpha}}_{s}-\hat{\alpha}^i_{s}) d s \right]\right)^2\nonumber\\
&&+(\bar{X}_t-X^i_t)^2\left[\frac{4}{N}A_1(E_1(t,0)+E_0(t))^2+\frac{2q}{N}(E_1(t,0)+E_0(t))
\right]\nonumber\\
&&+(\bar{X}_t-X^i_t)\left[\frac{2}{N}\sum\limits_{j\neq i}\alpha^j_t(E_1(t,0)+E_0(t))\right]\nonumber\\
&&+(\bar{X}_t-X^i_t)\left(\int_{t-\tau}^t (E_2(t,s-t,0)+E_1(t,s-t)(\bar\alpha_s-\alpha_s^i)ds\right)\left[
\frac{8}{N}A_1(E_1(t,0)+E_0(t))+\frac{2q}{N}\right]\nonumber\\
&&+\left(\int_{t-\tau}^t (E_2(t,s-t,0)+E_1(t,s-t)(\bar\alpha_s-\alpha_s^i)ds\right)\left[
\frac{2}{N}\sum\limits_{j\neq i}\alpha^j_t\right]\nonumber\\
&&+\left(\int_{t-\tau}^t (E_2(t,s-t,0)+E_1(t,s-t)(\bar\alpha_s-\alpha_s^i)ds\right)^2\left[
\frac{4}{N}A_1\right]
\Bigg\} dt .
\ea
Now, assuming that the players $j\neq i$ are using the strategies $\hat\alpha^j_t$ given by (\ref{Nash}), the quantity $\sum\limits_{j\neq i}\alpha^j_t$ becomes 
\ba
&&\sum_{j\neq i}\hat\alpha_t^j=-2A_1\left[\left(E_1(t,0)+E_0(t)+\frac{q}{2A_1}\right)(\bar{X}_{t}-X^i_{t}) 
\right.
\nonumber \\
&&\hskip 4cm \left.
+\int_{t-\tau}^{t}\left[E_2(t,s-t,0)+E_1(t,s-t)\right](\bar{\hat{\alpha}}_{s}-\hat{\alpha}^i_{s}) d s \right].\nonumber
\ea
Plugging this last expression in (\ref{duadequ3}), one sees that the terms after the square cancel and we get
\ba
\label{duadequ4}
&& - V^i(0,\xi^i,\alpha_{[0)})+J^i(0,\xi^i, \alpha_{[0)},(\alpha^i,\hat\alpha^{-i})) =  \nonumber \\
&&\mathbb{E} \int\limits_{0}^{T}\Bigg\{ 
\frac{1}{2}\left({\alpha}^i_t-2A_1\left[\left(E_1(t,0)+E_0(t)+\frac{q}{2A_1}\right)(\bar{X}_{t}-X^i_{t})
\right.\right.
\nonumber \\
&&\hskip 4cm \left.\left.
 +\int_{t-\tau}^{t}\left[E_2(t,s-t,0)+E_1(t,s-t)\right](\bar{\hat{\alpha}}_{s}-\hat{\alpha}^i_{s}) d s \right]\right)^2\Bigg\} dt .\nonumber\\
\ea
Consequently $V^i(0,\xi^i,\alpha_{[0)}) \leq J^i(0,\xi^i, \alpha_{[0)},(\alpha^i,\hat\alpha^{-i}))$, and choosing $\alpha^i=\hat\alpha^i$ leads to
\\
  $V^i(0,\xi^i,\alpha_{[0)})= J^i(0,\xi^i, \alpha_{[0)},(\hat\alpha^i,\hat\alpha^{-i}))$.
\end{proof}

\section{Financial Implications and Numerical Illustration}\label{FI}

The main finding is that taking into account repayment with delay  does not change the fact that the central bank providing liquidity is acting as a {\it clearing house} in all the Nash equilibria we identified (open-loop in Section \ref{Systemic-Risk-SDDE} or closed-loop in Sections  \ref{sec:HJB} and  \ref{caratheodory}).

The delay time, that is the single repayment maturity $\tau$ that we considered in this paper, controls the liquidity provided by borrowing and lending. The two extreme case are:
\begin{enumerate}
\item
No borrowing/lending: $\tau=0$:

In that case, no liquidity is provided and the log-reserves $X^i_t$ follow independent Brownian motions.
\item
No repayment: $\tau\geq T$:

This is the case studied previously in \cite{R.Carmona2013} and summarized in Section \ref{sec:nodelay}.
The rate of liquidity (the speed at which money is flowing through the system) is given by $\left[q+(1-\frac{1}{N})\phi_t\right]$ as shown in equation (\ref{ol-dXi2}).

\item
Intermediate regime $0<\tau<T$:

We conjecture that the rate of liquidity is monotone in $\tau$. For instance, in the case of the close-loop equilibrium obtained in Section \ref{caratheodory} given by (\ref{Nash}), the rate of liquidity is $\left[2E_1(t,0)+2E_0(t)+q\right]$ where the function $E_1$ and $E_0$ are  solutions to the system (\ref{eq1-1}--\ref{eq1-3}). These solutions are not given by closed form formulas. We computed them numerically. We show in Figure \ref{liquidity} that as expected, liquidity increases as $\tau$ increases. This is clear for values of $\tau$ which are small relative to the time horizon $T$. For values of $\tau$ which are large and comparable with $T$, the boundary effect becomes more important as oscillations propagate backward.
\end{enumerate}

\begin{figure}[H]
\includegraphics[width=15cm,height=8cm]{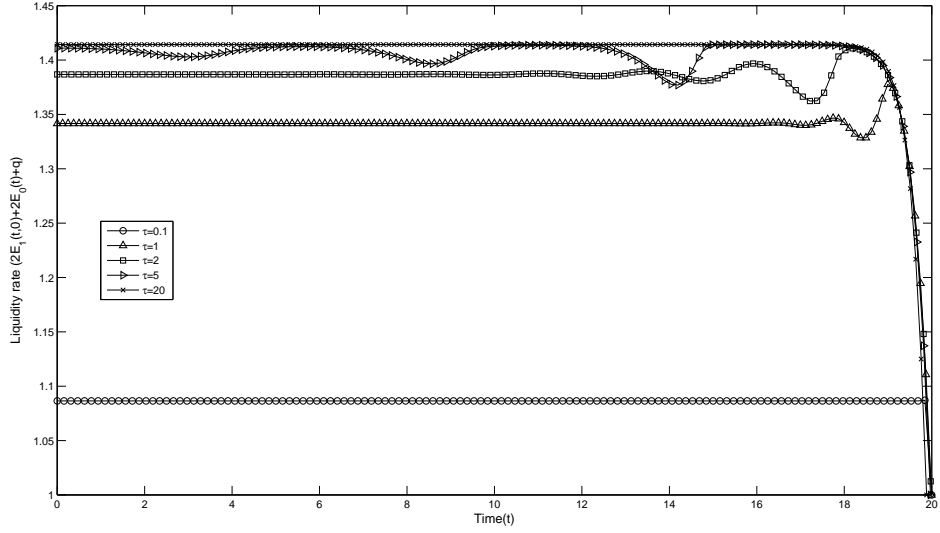}
\caption{Liquidity as a function of the delay time $\tau$. The parameters are $T=20$, $q=1$, $\eps=2$, and $c=0$.}
\label{liquidity}
\end{figure}

\section*{Acknowledgement}
The authors would like to thank  Romuald Elie and Phillip Yam for conversations on this subject.

\appendix

\section{Proof of Lemma \ref{thm-2}}\label{FABSDEproof}

\begin{proof}
Assuming that $(\check{X},\check{Y}, (\check{Z}^{k})_{k=1,\cdots,N})$ is given as an input, we solve the system \eqref{fo:Peng_Wu_system} for $\lambda=\lambda_0$ and the processes $\phi_t$, $\psi^k_t$, $r_t$ and  the random variable $\zeta$ replaced according to the prescriptions:
\ban
\phi_t&\leftarrow&\phi_t +\kappa\bigl[ \check Y_t- <\widetilde{\check{Y}}_{[t]} + q\check{X}_{[t]},\theta>\bigr]
\\
\psi^{k}_t&\leftarrow&\psi^{k}_t + \kappa\bigl[\check Z^{k}_t+\sigma(\frac{1}{N}-\delta_{i,k})\bigr],\quad k=1,\cdots,N\\
r_t&\leftarrow&r_t + \kappa\bigl[\check{X}_t+\bigl(1-\frac{1}{N}\bigr)\bigl[q\widetilde{\check{Y}}_t+ \left(q^2-\epsilon\right)\check{X}_t\bigr]\bigr]\\
\zeta&\leftarrow&\zeta + \kappa\bigl[-\check{X}_T+ c(1-\frac{1}{N})\check{X}_T\bigr],
\ean
and denote the solution by $({X},{Y}, ({Z}^{k})_{k=1,\cdots,N})$.
In this way, we defined a mapping 
$$
\Phi:(\check{X},\check{Y}, (\check{Z}^{k})_{k=1,\cdots,N})\rightarrow \Phi(\check{X},\check{Y}, (\check{Z}^{k})_{k=1,\cdots,N})=({X},{Y}, ({Z}^{k})_{k=1,\cdots,N}),
$$ 
and the proof consists in proving that the latter is a contraction for small enough $\kappa>0$.

Consider $(\widehat{X},\widehat{Y},(\widehat{Z}^{k})_{k=1,\cdots,N})=({X}-{X}^\prime,{Y}-{Y}^\prime, (Z^{k}-{Z}^{k\prime})_{k=1,\cdots,N})$ where $(X,Y,(Z^{k})_{k=1,\cdots,N})$ and $({X}^\prime,Y^\prime,({Z^{k}}^\prime)_{k=1,\cdots,N})$ are the corresponding image using inputs $(\check{X},\check{Y},(\check{Z}^{k})_{k=1,\cdots,N})$ and $({\check{X}^{\prime}},{\check{Y}^{\prime}},(\check{Z}^{k\prime})_{k=1,\cdots,N})$. 
We obtain  
\ba
\nonumber d \widehat{X}_t&=&\bigl[-(1-\lambda_0)\widehat{Y}_t
-\lambda_0<\widetilde{\widehat{Y}}_{[t]}+q\widehat{X}_{[t]},\theta> + \kappa\bigl[\widehat{\check Y_t}-<\widetilde{\widehat{\check{Y}}}_{[t]}+q\widehat{\check{X}}_{[t]},\theta>\bigr]\bigr]dt
\\
\nonumber &&+\sum_{k=1}^N[-(1-\lambda_0) \widehat Z^{k}_t +\kappa\widehat{\check Z}^{k}_t\bigr]dW^k_t\\
\nonumber d\widehat{Y}_t&=&\bigl[-(1-\lambda_0)\widehat{X}_t+\lambda_0\bigl(1-\frac{1}{N}\bigr)\bigl[q\widetilde{\widehat{Y}}_t+ (q^2-\epsilon)\widehat{X}_t\bigr] 
+\kappa\bigl[\widehat{\check{X}}_t + \bigl(1-\frac{1}{N}\bigr)\bigl[q\widetilde{\widehat{\check{Y}}}_t+ (q^2-\epsilon)\widehat{\check{X}}_t\bigr]\bigr]\bigr]dt\\
&&+\sum_{k=1}^N\widehat Z^{k}_tdW^k_t,
\ea
with initial condition $\widehat{X}_0=0$ and terminal conditions $ \widehat{Y}_T=(1-\lambda_0)\widehat{X}_T+\lambda_0c\bigl(1-\frac{1}{N}\bigr)\widehat{X}_T-\kappa\widehat{\check{X}}_T+\kappa c(1-\frac{1}{N})\widehat{\check{X}}_T$ and $\widehat{Y}_t=0$ for $t\in(T,T+\tau]$ in the case of $c >0$, and $\widehat Y_T=0$ and $\widehat{Y}_t=0$ for $t\in(T,T+\tau]$ in the case of $c=0$. As we stated in the text, we only give the proof in the case $c=0$ to simplify the notation. The proof of the case $c>0$ is a easy modification.
Using the form of the terminal condition and It$\hat{\mathrm{o}}$'s formula, we get

\ba
\nonumber &&0=\EE[\widehat{Y} _T\widehat{X} _T]\\
\nonumber &&=\EE\int_0^T\bigg\{
\widehat{Y}_t
\bigg[-(1-\lambda_0)\widehat{Y}_t
-\lambda_0<\widetilde{\widehat{Y}}_{[t]}+q\widehat{X}_{[t]},\theta> + \kappa\bigl[\widehat{\check Y_t}-<\widetilde{\widehat{\check{Y}}}_{[t]}+q\widehat{\check{X}}_{[t]},\theta>\bigr]
\bigg]\\
\nonumber &&\hskip25pt+ \widehat{X}_t
\bigg[
-(1-\lambda_0)\widehat{X}_t+\lambda_0\bigl(1-\frac{1}{N}\bigr)\bigl[q\widetilde{\widehat{Y}}_t+ (q^2-\epsilon)\widehat{X}_t\bigr]
+\kappa\bigl[\widehat{\check{X}}_t + \bigl(1-\frac{1}{N}\bigr)\bigl[q\widetilde{\widehat{\check{Y}}}_t+ (q^2-\epsilon)\widehat{\check{X}}_t\bigr]\bigr] 
\bigg]\\
&&\hskip45pt-(1-\lambda_0)\sum_{k=1}^N|\widehat Z^{k}_t|^2+\kappa\sum_{k=1}^N\widehat Z^{k}_t\widehat{\check Z}^{k}_t
\bigg\}dt\\
\nonumber &&=-(1-\lambda_0)\EE\int_0^T|\widehat{Y}_t|^2dt
-\lambda_0\EE\int_0^T\widehat{Y}_t<\widetilde{\widehat{Y}}_{[t]}+q\widehat{X}_{[t]},\theta>dt
+ \kappa\EE\int_0^T\widehat{Y}_t  \bigl[\widehat{\check Y_t}-<\widetilde{\widehat{\check{Y}}}_{[t]}+q\widehat{\check{X}}_{[t]},\theta>\bigr]dt\\
\nonumber &&\hskip 75pt-(1-\lambda_0)\EE\int_0^T|\widehat{X}_t|^2dt
+\lambda_0\bigl(1-\frac{1}{N}\bigr)\EE\int_0^T\widehat{X}_t\bigl[q\widetilde{\widehat{Y}}_t+ (q^2-\epsilon)\widehat{X}_t\bigr]dt\\
\nonumber &&\hskip 100pt+\kappa\EE\int_0^T\widehat{X}_t\bigl[\widehat{\check{X}}_t + \bigl(1-\frac{1}{N}\bigr)\bigl[q\widetilde{\widehat{\check{Y}}}_t+ (q^2-\epsilon)\widehat{\check{X}}_t\bigr]
\bigr]dt\\
&&\hskip 100pt-(1-\lambda_0)\EE\int_0^T \sum_{k=1}^N|\widehat Z^{k}_t|^2dt+\kappa\sum_{k=1}^N\widehat Z^{k}_t\widehat{\check Z}^{k}_t
 dt
\ea
and rearranging the terms we find:
\ba
\nonumber &&(1-\lambda_0)\bigl[\EE\int_0^T|\widehat{X}_t|^2dt + \EE\int_0^T|\widehat{Y}_t|^2 dt + \EE\int_0^T\sum_{k=1}^N|\widehat Z^{k}_t|^2\;dt\bigr]\\
\nonumber&&=\kappa\EE\int_0^T \widehat{X}_t\widehat{\check{X}}_t dt
-\lambda_0\EE\int_0^T\widehat{Y}_t<\widetilde{\widehat{Y}}_{[t]}+q\widehat{X}_{[t]},\theta>dt
+ \kappa\EE\int_0^T\widehat{Y}_t \bigl[\widehat{\check Y_t}-<\widetilde{\widehat{\check{Y}}}_{[t]}+q\widehat{\check{X}}_{[t]},\theta>\bigr]dt\\
\nonumber &&\hskip 75pt
+\lambda_0\bigl(1-\frac{1}{N}\bigr)\EE\int_0^T\widehat{X}_t\bigl[q\widetilde{\widehat{Y}}_t+ (q^2-\epsilon)\widehat{X}_t\bigr]dt\\
\nonumber &&\hskip 100pt+\kappa \bigl(1-\frac{1}{N}\bigr)\EE\int_0^T\widehat{X}_t\bigl[q\widetilde{\widehat{\check{Y}}}_t+ (q^2-\epsilon)\widehat{\check{X}}_t\bigr]
\bigr]dt+\kappa\EE\int_0^T\sum_{k=1}^N\widehat Z^{k}_t\widehat{\check Z}^{k}_t \; dt
\ea
Letting 
$\mu= \epsilon(1-\frac{1}{N})-q^2(1-\frac{1}{2N})^2>0$,  we obtain:
\ba
\nonumber &&(1-\lambda_0+\lambda_0\mu)\EE\int_0^T|\widehat{X}_t|^2dt
+(1-\lambda_0)\EE\int_0^T|\widehat Y_t|^2dt+(1-\lambda_0)\EE\int_0^T\sum_{k=1}^N|\widehat Z^{k}_t|^2dt\\
\nonumber&&\le  
\kappa\EE\int_0^T\widehat{Y}_t \bigl[\widehat{\check Y_t}-<\widetilde{\widehat{\check{Y}}}_{[t]}+q\widehat{\check{X}}_{[t]},\theta>\bigr]dt\\
\nonumber &&\hskip 15pt
+\kappa\bigl(1-\frac{1}{N}\bigr)\EE\int_0^T\bigg(\left(q^2-\epsilon\right)\widehat{\check{X}}_t+q\widetilde{\widehat{\check{Y}}}_t\bigr)\widehat{X}_tdt
\label{system-main}
+\kappa\EE\int_0^T\sum_{k=1}^N\widehat{Z}^{k}_t\widehat{\check{Z}}^{k}_tdt,
\ea
and a straightforward computation using repeatedly  Cauchy--Schwarz and Jensen's inequalities leads to the existence of a positive constant $K_1$ such that
\ba
\nonumber && (1-\lambda_0+\lambda_0\mu)\EE\int_0^T|\widehat{X}_t|^2dt
+(1-\lambda_0)\EE\int_0^T|\widehat Y_t|^2dt+(1-\lambda_0)\EE\int_0^T\sum_{k=1}^N|\widehat Z^{k}_t|^2dt\\
\nonumber &&\hskip 25pt\le\kappa K_1\bigg\{ \EE\int_0^T|\widehat{X}_t|^2dt+\EE\int_0^T|\widehat{Y}_t|^2dt+\EE\int_0^T\sum_{k=1}^N|\widehat{Z}^{k}_t|^2dt\\
\nonumber&&\hskip 45pt+ \EE\int_0^T|\widehat{\check{X}}_t|^2dt+\EE\int_0^T|\widehat{\check{Y}}_{t}|^2dt+\EE\int_0^T\sum_{k=1}^N|\widehat{\check{Z}}^{k}_t|^2dt\bigg\}.
\ea    
Referring to \cite{Bensoussan2015}, applying It$\hat{\mathrm{o}}$'s formula to  $|\widehat X_t|^2$ and $|\widehat Y_t|^2$,  Gronwall's inequality, and again Cauchy-Schwarz and Jensen's inequalities,  owing to $0\leq\lambda_0\leq 1$, we obtain a constant $K_2>0$ independent of $\lambda_0$ so that 
\be\label{cond-X}\nonumber
\sup_{0\leq t\leq T}\EE|\widehat X_t|^2\leq\kappa K_2 \left\{ \EE\int_0^T|\widehat{\check{X}}_t|^2+|\widehat{\check{Y}}_t|^2+\sum_{k=1}^N|\widehat{\check{Z}}^{k}_t|^2dt\right\}+K_2\left\{\EE\int_0^T|\widehat{Y}_t|^2+\sum_{k=1}^N|\widehat{Z}^{k}_t|^2dt\right\},
\en
\be\label{cond-int-X}\nonumber
\EE\int_0^T|\widehat{X}_t|^2dt\leq \kappa K_2T\left\{ \EE\int_0^T|\widehat{\check{X}}_t|^2+|\widehat{\check{Y}}_t|^2+\sum_{k=1}^N|\widehat{\check{Z}}^{k}_t|^2dt\right\}+K_2T\left\{\EE\int_0^T|\widehat{Y}_t|^2+\sum_{k=1}^N|\widehat{Z}^{k}_t|^2dt\right\},
\en
\ba\label{cond-BSDE}
\EE\int_0^T|\widehat{Y}_t|^2+\sum_{k=1}^N|\widehat{Z}^{k}_t|^2dt&\leq&\kappa K_2\left\{ \EE\int_0^T|\widehat{\check{X}}_t|^2+|\widehat{\check{Y}}_t|^2+\sum_{k=1}^N|\widehat{\check{Z}}^{k}_t|^2dt\right\} \nonumber \\
&&+ K_2\EE\int_0^T|\widehat{{X}}_t|^2dt . 
\ea
By using \eqref{cond-BSDE}, there exists $0<\mu'<\mu/{K_2}$ such that 
\ba
&&\lambda_0\mu'K_2\EE\int_0^T|\widehat{{X}}_t|^2dt \nonumber\\
&&\hskip 25pt
\geq\lambda_0\mu'\left(\EE\int_0^T|\widehat{Y}_t|^2+\sum_{k=1}^N|\widehat{Z}^{k}_t|^2dt\right)
-\lambda_0\mu' \kappa K_2\left\{ \EE\int_0^T|\widehat{\check{X}}_t|^2+|\widehat{\check{Y}}_t|^2+\sum_{k=1}^N|\widehat{\check{Z}}^{k}_t|^2dt\right\}\nonumber \\
&&\hskip 25pt
\geq\lambda_0\mu'\left(\EE\int_0^T|\widehat{Y}_t|^2+\sum_{k=1}^N|\widehat{Z}^{k}_t|^2dt\right)
-\mu' \kappa K_2\left\{ \EE\int_0^T|\widehat{\check{X}}_t|^2+|\widehat{\check{Y}}_t|^2+\sum_{k=1}^N|\widehat{\check{Z}}^{k}_t|^2dt\right\}\nonumber \\
\ea
Therefore, we have
\ba
\nonumber && \bigg(1-\lambda_0+\lambda_0(\mu-K_2\mu')\bigg)\EE\int_0^T|\widehat{X}_t|^2dt\\
\nonumber &&+(1-\lambda_0+\lambda_0\mu')\EE\int_0^T|\widehat Y_t|^2dt+(1-\lambda_0+\lambda_0\mu')\EE\int_0^T\sum_{k=1}^N|\widehat Z^{k}_t|^2dt\\
\nonumber &\leq&\kappa K_1\bigg\{ \EE\int_0^T|\widehat{X}_t|^2dt+\EE\int_0^T|\widehat{Y}_t|^2dt+\EE\int_0^T\sum_{k=1}^N|\widehat{Z}^{k}_t|^2dt\\
\nonumber&&\quad\quad +\EE\int_0^T|\widehat{\check{X}}_t|^2dt+\EE\int_0^T|\widehat{\check{Y}}_{t}|^2dt+\EE\int_0^T\sum_{k=1}^N|\widehat{\check{Z}}^{k}_t|^2\bigg\}\\
\nonumber && +\kappa K_2\mu'\left\{ \EE\int_0^T|\widehat{\check{X}}_t|^2dt+\EE\int_0^T|\widehat{\check{Y}}_t|^2dt+\EE\int_0^T\sum_{k=1}^N|\widehat{\check{Z}}^{k}_t|^2dt\right\} .\label{cond-main}\\
\ea
Note that since $\mu-K_2\mu'$ and $\mu'$ stay in positive, we have  $(1-\lambda_0+\lambda_0(\mu-K_2\mu'))\geq \mu''$ and $(1-\lambda_0+\lambda_0\mu')\geq \mu''$ where for some $\mu''>0$.
Combining the inequalities (\ref{cond-X}-\ref{cond-main}), we obtain 
\ba
&&\nonumber \EE\int_0^T|\widehat{X}_t|^2dt+\EE\int_0^T|\widehat{Y}_t|^2dt+\EE\int_0^T\sum_{k=1}^N|\widehat{Z}^{k}_t|^2dt\nonumber\\
&&\hskip 75pt\le \kappa K \left( \EE\int_0^T|\widehat{\check{X}}_t|^2dt+\EE\int_0^T|\widehat{\check{Y}}_{t}|^2dt+\EE\int_0^T\sum_{k=1}^N|\widehat{\check{Z}}^{k}_t|^2dt\right),\nonumber \\
\ea
where the constant $K$ depends upon $\mu'$, $\mu''$, $K_1$, $K_2$, and $T$. 
Hence, $\Phi$ is a strict contraction for sufficiently small $\kappa$.  
\end{proof}

\bibliographystyle{plainnat}
\bibliography{references-delay-games}

\end{document}